\documentclass{theoretics}

\title{Positivity-hardness results  on Markov decision processes}

\ThCSauthor{Jakob Piribauer}{jakob.piribauer@tu-dresden.de}[0000-0003-4829-0476]
\ThCSauthor{Christel Baier}{christel.baier@tu-dresden.de}[0000-0002-5321-9343]

\ThCSaffil{Technische Universit{\"a}t Dresden}

\ThCSshortnames{J.\ Piribauer and C.\ Baier}
\ThCSshorttitle{Positivity-hardness results  on Markov decision processes}
\ThCSyear{2024}
\ThCSarticlenum{9}
\ThCSreceived{Mar 6, 2023}
\ThCSrevised{Oct 11, 2023}
\ThCSaccepted{Feb 23, 2024}
\ThCSpublished{Apr 7, 2024}
\ThCSdoicreatedtrue

\ThCSkeywords{Positivity problem for linear recurrence sequence,
Skolem problem,
Markov decision processes
}
\ThCSthanks{This paper is an extension of work published at ICALP 2020 \cite{icalp2020}. It extends the conference version by the results for one-counter MDPs, energy objectives, quantiles, and cost problems. 
These additional results are also presented in the PhD thesis \cite{piribauer2021}.
Furthermore, full proofs omitted in the conference version and a detailed and improved description of all constructions are provided. 
\\

\textbf{Funding}
This work was partly funded by the DFG Grant 389792660 as part of TRR 248 (Foundations of Perspicuous Software Systems).}

\usepackage{tikz}
\usetikzlibrary{arrows.meta}
\usetikzlibrary{positioning,automata,fit,shapes,calc}
\usepackage{pgfplots}
\pgfplotsset{width=10cm,compat=1.9}

\newcommand{\PE}{\mathbb{PE}}
\newcommand{\CE}{\mathbb{CE}}

\newcommand{\VaR}{\mathit{VaR}}
\newcommand{\CVaR}{\mathit{CVaR}}

\newcommand{\cC}{\mathcal{C}}

\newcommand{\cG}{\mathsf{G}}

\newcommand{\cI}{\mathcal{I}}

\newcommand{\cM}{\mathcal{M}}
\newcommand{\cN}{\mathcal{N}}
\newcommand{\cO}{\mathcal{O}}
\newcommand{\cP}{\mathcal{P}}

\newcommand{\cT}{\mathcal{T}}

\newcommand{\eqdef}{\ensuremath{\stackrel{\text{\tiny def}}{=}}}

\renewcommand{\Pr}{\mathrm{Pr}}

\newcommand{\AP}{\mathsf{AP}}
\newcommand{\sinit}{s_{\mathit{\scriptscriptstyle init}}}

\newcommand{\Act}{\mathit{Act}}

\newcommand{\last}{\mathit{last}}
\newcommand{\length}{\mathit{length}}
\newcommand{\Cyl}{\mathit{Cyl}}
\newcommand{\sched}{\mathfrak{S}}
\newcommand{\tsched}{\mathfrak{T}}

\newcommand{\wgt}{\mathit{wgt}}
\newcommand{\act}{\alpha}
\newcommand{\Size}{\mathit{size}}
\newcommand{\goal}{\mathit{goal}}
\newcommand{\fail}{\mathit{fail}}

\newcommand{\trap}{\mathit{trap}}
\newcommand{\Goal}{\mathit{Goal}}

\newcommand{\fpath}{\pi}

\newcommand{\rawdiaplus}{%
  \begin{tikzpicture}
    \useasboundingbox (-0.7ex, -0.9ex) rectangle (0.7ex, 0.9ex);
    \node (w) at (-0.7ex,0) {};
    \node (e) at (+0.7ex,0) {};
    \node (s) at (0,-0.9ex) {};
    \node (n) at (0,+0.9ex) {};
    \draw (n.center) -- (e.center) -- (s.center) -- (w.center) -- (n.center);
    \draw (n.center) -- (s.center);
    \draw (e.center) -- (w.center);
  \end{tikzpicture}}

\newcommand{\rawdiaminus}{%
  \begin{tikzpicture}
    \useasboundingbox (-0.7ex, -0.9ex) rectangle (0.7ex, 0.9ex);
    \node (w) at (-0.7ex,0) {};
    \node (e) at (+0.7ex,0) {};
    \node (s) at (0,-0.9ex) {};
    \node (n) at (0,+0.9ex) {};
    \draw (n.center) -- (e.center) -- (s.center) -- (w.center) -- (n.center);
    \draw (e.center) -- (w.center);
  \end{tikzpicture}}

\ThCSnewtheostd{assumption}

\begin{document}

\maketitle

\begin{abstract}

This paper investigates a series of optimization problems for one-counter Markov decision processes (MDPs) and integer-weighted MDPs with finite state space. Specifically, it considers problems addressing termination probabilities and expected termination times for one-counter MDPs, as well as satisfaction probabilities of energy objectives, conditional and partial expectations,  satisfaction probabilities of constraints on the total accumulated weight, the computation of quantiles for the accumulated weight, and  the conditional value-at-risk for accumulated weights for integer-weighted MDPs.  Although algorithmic results are available for some special instances, the decidability status of the decision versions of these problems is unknown in general.

The paper demonstrates that these optimization problems are inherently mathematically difficult by providing polynomial-time reductions from the Positivity problem for linear recurrence sequences. This problem is a well-known number-theoretic problem whose decidability status has been open for decades and it is known that  decidability of the Positivity problem would have far-reaching consequences in analytic number theory. So, the reductions presented in the paper show that an algorithmic solution to any of the investigated problems is not possible without a major breakthrough in analytic number theory.
The reductions rely on the construction of MDP-gadgets that encode the initial values and linear recurrence relations of linear recurrence sequences. 
These gadgets can flexibly be adjusted to prove the various Positivity-hardness results.

\end{abstract}

\section{Introduction}
\label{chap:intro}

When modelling and analyzing computer systems and their interactions with their environment, two qualitatively different kinds of uncertainty about the evolution of the system execution play a central role: non-determinism and probabilism.
If a system is, for example, employed in an unknown environment or depends on user inputs or concurrent processes, modelling the system as non-deterministic accounts for all possible external influences, sequences of user inputs, or possible orders in which concurrent events take place.  If transition probabilities between the states of a system, such as the failure probability of components or the probabilities in a  probabilistic choice employed in a randomized algorithm, are known or can be estimated, it is appropriate to model this behavior as probabilistic. A pure worst- or best-case analysis is not very informative in such cases and the additional probabilistic information available should be put to use.
\emph{Markov decision processes (MDPs)} are a standard operational model combining non-deterministic and probabilistic behavior and are widely used in operations research, artificial intelligence, and verification among others.

In each state of an MDP, there is a non-deterministic choice from a set of actions. Each action specifies a probability distribution over the possible successor states according to which a transition is chosen randomly.
Typical optimization problems on MDPs  require resolving the non-deterministic choices by specifying a \emph{scheduler} such that a quantitative objective function is optimized.
For example, the standard model-checking problem asks for the minimal or maximal probability  that an execution satisfies a given linear-time property. Here, minimum and maximum range over all resolutions of the non-deterministic choices, i.e., over all schedulers. 
This model-checking problem is known to be 2EXPTIME-complete if the property is given in linear temporal logic (LTL) \cite{CY95} and solvable in polynomial time if the property is given by a deterministic automaton \cite{deAlfaro1999,BaierKatoen08}.
Many quantitative aspects of a system can be modeled by equipping an MDP with weights that are collected in each step. These weights might represent time, energy consumption, utilities, or  generally speaking any sort of costs or rewards incurred. Classical optimization problems in this context that are known to be solvable in polynomial time include the optimization of the expected value of 
the total accumulated weight before a target state is reached, the so-called \emph{stochastic shortest path problem} (SSPP) \cite{bertsekas1991,deAlfaro1999,lics2018}, the expected value of the reward earned on average per step, the so-called expected \emph{mean payoff} or \emph{long-run average}, or the expected \emph{discounted accumulated weight} where after each step a discount factor is applied to all future weights (for the latter two, see, e.g., \cite{hordijk1979,puterman1994}). 

Of course, there is a vast landscape of further optimization problems on finite-state MDPs that have been analyzed.
We are, nevertheless, not aware of natural decision problems for standard (finite-state) MDPs with a single weight function and single objective that are known to be undecidable.
 Undecidability results have been established for more expressive models. This applies, e.g., to recursive MDPs \cite{etessami2015}, MDPs with two or more weight functions \cite{BKKW14,randour2017}, or partially observable MDPs \cite{madani1999,BaiGroeBer12}. 

 In this paper, we will investigate a series of optimization problems that have been studied in the literature, but  are open in general.
We will show that these problems possess an inherent mathematical difficulty that makes  algorithmic solutions impossible without a major breakthrough in analytic number theory.
Formally, this result is obtained by reductions from the \emph{Positivity problem} for \emph{linear recurrence sequences}, a number theoretic problem whose decidability status has been open for many decades (see, e.g., \cite{halava2005,ouaknine2012decision,ouaknine2015linear}).

\subsection{{Positivity problem}}
\begin{definition}[Positivity problem]
The Positivity problem for linear recurrence sequences asks whether such a sequence stays non-negative. More formally,
given a natural number $k\geq 2$, and rationals $\alpha_i$ and $\beta_j$ with $1\leq i \leq k$ and $0\leq j \leq k-1$, let $(u_n)_{n\geq0}$ be defined  by the initial values $u_0=\beta_0$, \dots, $u_{k-1}=\beta_{k-1}$ and the linear recurrence relation
\[ u_{n+k} = \alpha_1 u_{n+k-1} + \dots + \alpha_k u_n \]
for all $n\geq 0$. The Positivity problem asks to decide whether $u_n \geq  0$ for all $n$.\footnote{We do not distinguish between the Positivity problem and its complement in the sequel. So, we also refer to the problem whether there is an $n$ such that $u_n<0$ as the Positivity problem.} The number $k$ is called the order of the linear recurrence sequence.
\end{definition}

The Positivity problem is closely related to  the famous \emph{Skolem problem}. The Skolem problem asks whether there is an $n$ such that $u_n= 0$ for a given linear recurrence sequence $(u_n)_{n\geq 0}$.  It is well-known that the 
 Skolem problem is polynomial-time reducible to the Positivity problem (see, e.g., \cite{everest2003}). 
The Positivity problem and the Skolem problem are outstanding problems in the fields of number theory and theoretical computer science (see, e.g., \cite{halava2005,ouaknine2012decision,ouaknine2015linear}) and their decidability has been open for many decades. 
Deep results establish decidability for both problems for linear recurrence sequences of low order or for restricted classes of sequences 
\cite{shorey1984distance,vereshchagin1985problem,ouaknine2014positivity,ouaknine2014,ouaknine2014ultimate}. A proof of decidability or undecidability of the Positivity problem for arbitrary sequences, however, withstands all known number-theoretic techniques. In \cite{ouaknine2014}, it is shown that decidability of the Positivity problem (already for linear recurrence sequences of order $6$) would entail a major breakthrough 
in the field of Diophantine approximation of transcendental numbers, an area of analytic number theory.

We call a problem to which the Positivity problem is reducible \emph{Positivity-hard}. 
From a complexity theoretic point of view, the Positivity problem is known to be at least as hard as the decision problem for the universal fragment of the theory of the reals with addition, multiplication, and order \cite{ouaknine2014ultimate}, a problem known to be coNP-hard and to lie in PSPACE \cite{canny1988some}. As most of the problems we will address are PSPACE-hard, the reductions in this paper do not provide new lower bounds on the computational complexity.
The hardness results in this paper hence refer to the far-reaching consequences on major open problems that a decidability result would imply.
Furthermore, of course, the undecidability of the Positivity problem would entail the undecidability of any Positivity-hard problem.

\subsection{Problems under investigation and related work on these problems}
In the sequel, we briefly describe the problems studied in this paper and describe related work on these problems. In general, the decidability status of all of these problems is open and we will prove them to be Positivity-hard.

\paragraph{Energy objectives, one-counter MDPs, and quantiles.}
If weights model a resource like energy that can be consumed and gained during a system execution, a natural problem is to determine the worst- or best-case probability that the system never runs out of the resource. This is known as the \emph{energy objective}. 
There has been work on combinations of the energy objective with further objectives such as parity objectives \cite{ChatDoy11,MaySchTozWoj17} and expected mean payoffs  \cite{BKN16}. Previous work on this objective focused on the possibility to satisfy the objective (or the combination of objectives) almost surely. 
The quantitative problem whether it is possible to satisfy an energy objective with probability greater than some threshold $p$ is open.

The complement of the energy objective can be found in the context of \emph{one-counter MDPs} (see \cite{brazdil2010,brazdil2011,brazdil2012}): Equipping an MDP with a counter that can be increased and decreased can be used to model a simple form of recursion and  can be seen as a special case of pushdown MDPs. The process is said to terminate as soon as the counter value drops below $0$ and the standard task is to compute maximal or minimal termination probabilities. In one-counter MDPs that terminate almost surely, one furthermore can ask for the extremal expected termination times, i.e. the expected number of steps until termination.
On the positive side, for one-counter MDPs, it is decidable in polynomial time whether there is a scheduler that ensures termination with probability $1$ \cite{brazdil2010}. Furthermore, \emph{selective termination}, which requires termination 
to occur inside a specified set of states  can be decided in exponential time \cite{brazdil2010}.
On the other hand, the computation of the optimal value and the quantitative decision problem whether the optimal value exceeds a threshold $p$ are left open in the literature. 
For selective termination, even the question whether the supremum of termination probabilities over all schedulers is $1$ is open.
Furthermore, also the problem to compute the minimal or maximal expected termination time of a one-counter MDP that terminates almost surely under any scheduler is open.
There are, however, approximation algorithms for the optimal termination probability \cite{brazdil2011}
and for the expected termination time of  almost surely terminating one-counter MDPs  \cite{brazdil2012}.
One-counter MDPs can be seen as a special case of 
recursive MDPs \cite{etessami2015}. For general recursive MDPs, the qualitative decision problem whether the maximal termination probability is $1$ is undecidable while for restricted forms, so-called 1-exit recursive MDPs, the qualitative and also the quantitative  problem is decidable in polynomial space \cite{etessami2015}. One-counter MDPs can be seen as a special case of 1-box recursive MDPs in the terminology of \cite{etessami2015}, a restriction orthogonal to 1-exit recursive MDPs.

The termination probability of one-counter MDPs and the satisfaction probability of the energy objective are closely related to the computation of \emph{quantiles} (see  \cite{UB13,baier2014energy,randour2017}).
Given a probability value $p$, here the task is to compute the best bound $b$ such that the maximal or minimal probability that the accumulated weight exceeds the bound is at most or at least $p$. The decision version whether the maximal or minimal probability that the accumulated weight before reaching a target state exceeds $b$ is at least or at most $p$ is also known as the \emph{cost problem} (see \cite{haase2015,haase2017computing,lics2018}).
The computation of quantiles and the cost problem have been addressed for MDPs with non-negative weights and  are solvable in exponential time in this setting
\cite{UB13,haase2015}. The decision version of the cost problem with non-negative weights is furthermore PSPACE-hard for a single inequality on the accumulated weight and EXPTIME-complete if a Boolean combination of inequality constraints on the accumulated weight is considered \cite{haase2015}.
For the setting with arbitrary weights, 
\cite{lics2018} provides solutions to the qualitative question whether a constraint on the accumulated weight is satisfied with probability $1$ (or $>0$). Further,
it is  known that the quantitative problem is undecidable if multiple objectives with multiple  weight functions have to be satisfied simultaneously \cite{randour2017}.

\paragraph{Non-classical stochastic shortest path problems (SSPPs).}
The classical SSPP described above requires that a goal state is reached almost surely. In many situations, however, there might be no schedulers reaching the target with probability $1$ or schedulers that miss the target with positive probability are of interest, too.
Two non-classical variants that drop this requirement are the conditional SSPP (see \cite{tacas2017,fossacs2019}) and the partial SSPP (see \cite{chen2013,chen2013prism}). In the conditional SSPP, the goal is to optimize the conditional expected accumulated weight before reaching the target under the condition that the target is reached. In other words, the average weight of all paths reaching the target has to be optimized.
In the partial SSPP, paths not reaching the target are not ignored, but assigned weight $0$.
Possible applications for these non-classical SSPPs include  the analysis of probabilistic programs where no guarantees on almost sure termination can be given (see, e.g., \cite{gretz2014,katoen2015,barthe2016,chatterjee2016,olmedo2018}), the analysis of fault-tolerant systems where error scenarios might occur with small, but positive probability, or the trade-off analysis with conjunctions of utility and cost constraints that are achievable with positive probability, but not almost surely (see, e.g., \cite{baier2014}).
In \cite{chen2013} and \cite{tacas2017}, partial and conditional expectations, respectively, have been addressed in the setting of non-negative weights. In both-cases, the optimal value can be computed in exponential time
\cite{chen2013,tacas2017} while the threshold problem is PSPACE-hard \cite{fossacs2019,tacas2017}. In MDPs with positive and negative weights, it is known that the optimal values might be irrational and that optimal schedulers might require infinite memory \cite{fossacs2019}.

Conditional expectations also play an important role for some risk measures. The \emph{conditional value-at-risk (CVaR)} is an established risk measure (see, e.g., \cite{Uryasev00,AcerbiTasche02}) defined as the conditional expected outcome under the condition that the outcome belongs to the $p$ worst outcomes for a given probability value $p$.
In the context of optimization problems on weighted MDPs, the CVaR has been studied  for mean-payoffs and weighted reachability where only one terminal weight is collected per run (see \cite{kretinsky2018}), and for the accumulated weight before reaching a target state in MDPs with non-negative weights (see \cite{ahmadi2021}).
The CVaR for accumulated weights can be optimized in MDPs with non-negative weights in exponential time \cite{icalp2020,Meggendorfer22}.

\subsection{{Contribution}}

We develop a technique to provide reductions from the Positivity problem to threshold problems on MDPs, asking whether the optimal value of a quantity \emph{strictly} exceeds a given rational threshold. The resulting reductions are based on the construction of MDP-gadgets that allow to encode the linear recurrence relation of a linear recurrence sequence and the initial values, respectively. 
The approach turns out to be quite flexible. By adjusting the gadgets encoding initial values, we can provide reductions of the same overall structure for several of the optimization problems we discussed. Through further chains of reductions depicted in Figure \ref{fig:overview_positivity}, we establish Positivity-hardness for the full series of optimization problems under investigation. The main result of this paper  consequently is the following:

 \paragraph{Main result.}
 The Positivity problem is polynomial-time reducible to the  threshold problems  for the optimal values of the following quantities:
 \begin{itemize}
 \item termination probabilities  of one-counter MDPs,
 \item  expected termination times of almost surely terminating one-counter MDPs,
  \item the satisfaction probabilities of  energy objectives in MDPs with weights in $\mathbb{Z}$,
  \item the probability to satisfy an inequality on the accumulated weight (cost problem) in MDPs with weights in $\mathbb{Z}$,
 \item conditional expectations (conditional SSPP) in MDPs with weights in $\mathbb{Z}$,
 \item partial expectations (partial SSPP) in MDPs with weights in $\mathbb{Z}$,
  \item conditional values-at-risk for accumulated weights (before reaching a goal) in MDPs with weights in $\mathbb{Z}$, and
 \item a two-sided version of partial expectations in MDPs with two non-negative weight functions with values in $\mathbb{N}$.
 \end{itemize}
 Furthermore, an algorithm for 
 \begin{itemize}
 \item the computation of quantiles for accumulated weights in MDPs with weights in $\mathbb{Z}$
 \end{itemize}
would imply the decidability of the Positivity-problem.

\begin{figure}[p]
  \begin{center}
    \includegraphics[width=\linewidth]{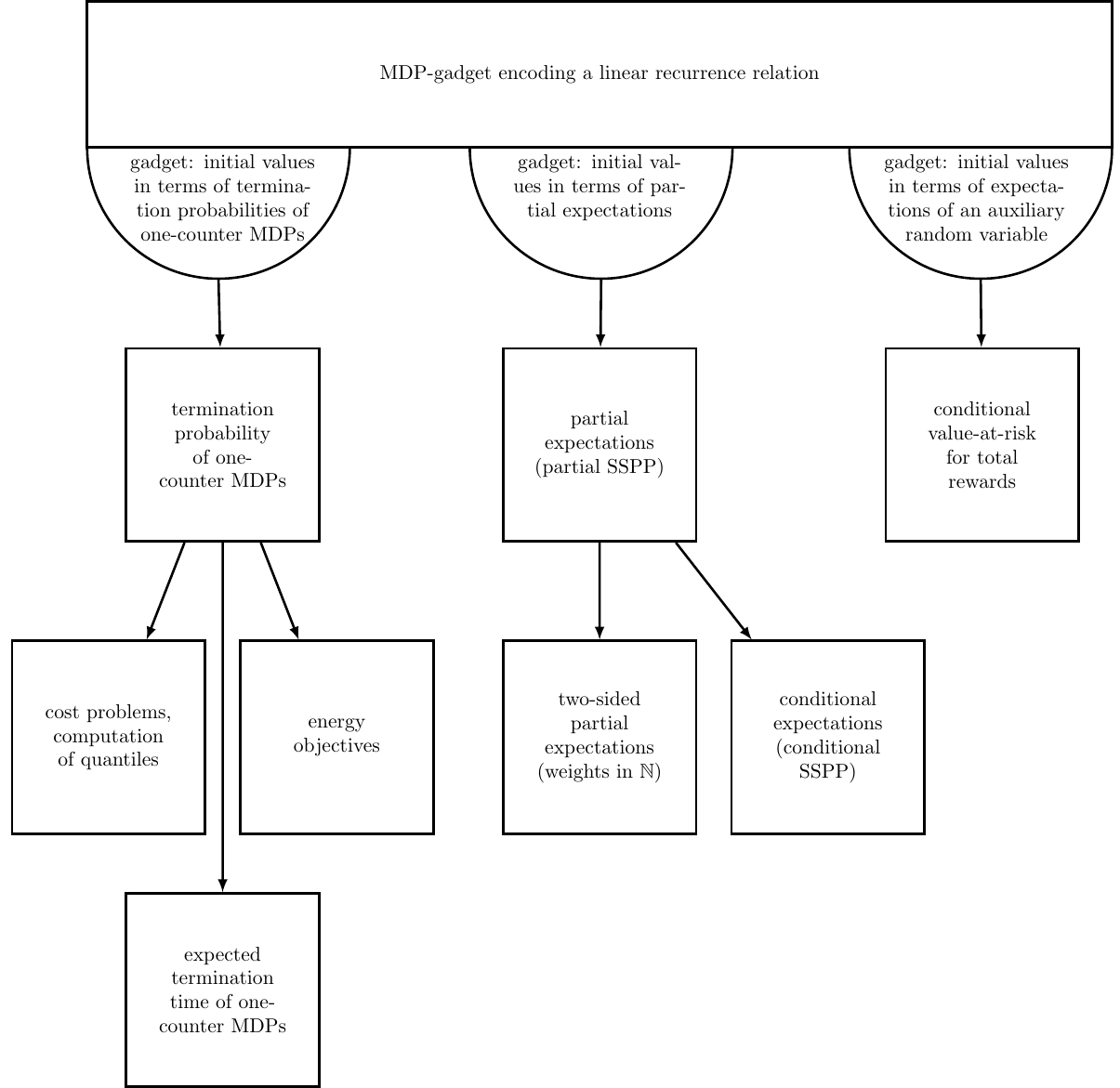}
\end{center}
\caption[Overview of the dependencies between the Positivity-hardness results.]{Overview of the dependencies between the Positivity-hardness results. The squares refer to the threshold problems for the respective quantities.}
\label{fig:overview_positivity}
\end{figure}

\subsection{Related work on Skolem- and Positivity-hardness in  verification}

In
\cite{akshay2015}, the Positivity-hardness of decision problems
for Markov chains has been established. The problems studied in \cite{akshay2015} are (1) to decide whether for given states $s$, $t$
and rational number~$p$, there is a positive integer $n$ such that the 
probability to reach $t$ from $s$ in $n$ steps  is at least~$p$,
and (2)
the model checking problem for a probabilistic variant of monadic logic 
and  a variant of LTL that treats Markov chains as linear transformers of
probability distributions. A  connection between similar problems  and the Skolem problem and Positivity problem has also been conjectured  in \cite{beauquier2006logic,agrawal2015approximate}.
These decision problems are of quite different nature than the
problems studied here. In particular, the problems are shown to be Positivity-hard in Markov chains. In contrast, e.g.,
partial and conditional expectations in Markov chains can be computed in polynomial time \cite{fossacs2019} and
the threshold problem for the termination probability of recursive Markov chains, which subsume one-counter Markov chains, can be solved in polynomial space \cite{DBLP:journals/jacm/EtessamiY09}. So, the Positivity-hardness of the corresponding problems on MDPs is not inherited from Positivity-hardness on Markov chains. Instead, our reductions show how the non-determinism in MDPs allows encoding linear recurrence sequences in terms of optimal values of various quantitative objectives by forcing an optimal scheduler to take certain decisions.
 Consequently,  the reductions are of a different nature than the reductions in \cite{akshay2015}. There, the behavior of a Markov chains in $n$ steps can directly be expressed by $P^n$ where $P$ is the transition probability matrix, which resembles the matrix formulation of the Positivity problem, which asks for a matrix $M$ and an initial  vector  $v$ whether there is an $n$ such that $M^n v $ lies within a half-space $H$.

In this context also the results of 
\cite{ChonOuakWor16} and \cite{MSS20} are remarkable
as they show the decidability, subject to Schanuel's conjecture, of 
reachability problems in
continuous linear dynamical systems and continuous-time MDPs, respectively,
as  instances of the continuous Skolem problem.
In other areas of formal verification, the Skolem problem and the Positivity problem play an important role in the context of the 
termination of linear programs \cite{ben2012termination,tiwari2004termination,braverman2006termination,ouaknine2015linear}.

The Positivity-hardness results leave the possibility open that the problems under consideration are undecidable.
Remarkable undecidability results in this context are  presented in  \cite{kaminski2015hardness}: 
The hardness of deciding almost sure termination and almost sure termination with finite expected termination time for purely probabilistic programs formulated in the probabilistic fragment of probabilistic guarded command language (pGCL) \cite{mciver2005abstraction} is pinpointed to levels of the arithmetical hierarchy (for details on the arithmetical hierarchy, see, e.g., \cite{odifreddi1992classical}).
The results reach up to $\Pi^0_3$-completeness for deciding universal almost sure termination with finite expected termination time ($\Pi_1^0$-complete problems are already undecidable while still co-recursively enumerable).  Undecidability is not surprising as the programs subsume ordinary programs. But the universal halting problem for ordinary programs is only $\Pi_2^0$-complete showing that deciding universal termination  with finite expected termination time of probabilistic programs is strictly harder. Similarly deciding termination from a given initial configuration is $\Sigma_1^0$-complete for ordinary programs (halting problem) while deciding almost sure termination with finite expected termination time for probabilistic programs from a given initial configuration is $\Sigma_2^0$-complete.
 Operational semantics of  pGCL-programs can be given as infinite-state MDPs \cite{gretz2014}. Applied to the purely probabilistic fragment, this leads to infinite-state Markov chains.

\subsection{Outline}

In the following Section \ref{sec:prelim}, we provide necessary definitions and present our notation.
In Section~\ref{sec:outline}, we outline the general structure of the gadget-based reductions from the Positivity-problem and construct 
 an MDP-gadget in which a linear recurrence relation can be encoded in terms of the optimal values for a variety of optimization problems (Section \ref{sec:gadget_recurrence}).
Afterwards, we  construct gadgets encoding also the initial values of a linear recurrence sequence and provide the reductions from the Positivity problems and all subsequent reductions as depicted in  Figure~\ref{fig:overview_positivity} (Section \ref{sec:reduction}).
We conclude with final remarks and an outlook on future work (Section~\ref{sec:conclusion}).


\section{Preliminaries}
\label{sec:prelim}

We assume some familiarity with Markov decision processes and briefly introduce our notation in the sequel. More details can be found in text books such as \cite{puterman1994}.

\paragraph{Markov decision process.}
A \emph{Markov decision process} (MDP) is a tuple $\mathcal{M} = (S,\Act,P,\sinit)$
where 
$S$ is a finite set of states,
$\Act$ is a finite set of actions,
$P \colon S \times \Act \times S \to [0,1] \cap \mathbb{Q}$ is the
transition probability function for which we require that
$\sum_{t\in S}P(s,\act,t) \in \{0,1\}$
for all $(s,\alpha)\in S\times \Act$, and
$\sinit \in S$ is the initial state.
Depending on the context, we enrich  MDPs with 
a weight function $\wgt \colon S \times \Act \to \mathbb{Z}$, 
a finite set of atomic propositions $\AP$ and a labeling function
$L\colon S\to 2^{\AP}$, or
a designated set of goal states $\Goal$.
The \emph{size} of an MDP $\cM$, denoted by $\Size(\cM)$,
is the sum of the number of states
plus the total sum of the  lengths of the encodings of the non-zero
probability values
$P(s,\alpha,s')$ as fractions of co-prime integers in binary and, if present, the  lengths of the encodings of the weight values $\wgt(s,\alpha)$ in binary.

We write $\Act(s)$ for the set of actions that are enabled in a state $s$,
i.e., $\act \in \Act(s)$ if and only if $\sum_{t\in S}P(s,\act,t) =1$. Whenever the process is in a state $s$, a non-deterministic choice between the enabled actions $\Act(s)$ has to be made.
We call a state \emph{absorbing} if the only enabled actions lead to the state itself with probability $1$ and weight $0$. If there are no enabled actions, we call a state \emph{terminal} or a \emph{trap state}.
The paths of $\cM$ are finite or
infinite sequences $s_0 \, \act_0 \, s_1 \, \act_1 \, s_2 \, \act_2 \ldots$
where states and actions alternate such that
$P(s_i,\act_i,s_{i+1}) >0$ for all $i\geq0$. Throughout this section, we  assume that all states are reachable from the initial state  in any MDP, i.e., that there is a finite path from $\sinit$ to each state $s$.
We extend the weight function to finite paths.
For a finite path $\fpath =
    s_0 \, \act_0 \, s_1 \, \act_1 \,  \ldots \act_{k-1} \, s_k$, 
we denote its accumulated weight  by
  \[\wgt(\fpath)=
   \wgt(s_0,\act_0) + \ldots + \wgt(s_{k-1},\act_{k-1}).\]
Similarly, we extend the transition probability function to finite paths and write
\[ P(\fpath) =
   P(s_0,\act_0,s_1) 
   \cdot \ldots \cdot P(s_{k-1},\act_{k-1},s_k).\]

   A \emph{one-counter MDP} is an MDP equipped with a counter. Each state-action pair increases or decreases the counter or leaves the counter unchanged. A one-counter MDP is said to terminate if the counter value drops below zero. We  view one-counter  MDPs as MDPs with a weight-function $\wgt\colon S\times \Act \to \{-1,0,+1\}$. In this formulation a one-counter MDP terminates when a prefix $\pi$ of a path satisfies $\wgt(\pi)<0$.

 A \emph{Markov chain} is an MDP in which the set of actions is a singleton. There are no non-deterministic choices in a Markov chain and hence we drop the set of actions. Consequently, a Markov chain is a tuple $\cM=(S,P,\sinit)$, possibly extended with a weight function, a labeling, or a designated set of goal states. The transition probability function $P$ is a function from $S\times S$ to $ [0,1]\cap \mathbb{Q}$ such that $\sum_{t\in S} P(s,t)\in \{0,1\}$ for all $s\in S$.

\paragraph{Scheduler.}
A \emph{scheduler} for an MDP $\cM=(S,\Act,P,\sinit)$
is a function $\sched$ that assigns to each finite path $\fpath$ not ending in trap state
a probability distribution over $\Act(\last(\fpath))$ where  $\last(\fpath)$ denotes the last state of $\fpath$. This probability distribution indicates which of the enabled actions is chosen with which probability under $\sched$ after the process has followed the finite path $\fpath$.

We allow schedulers to be \emph{randomized} and \emph{history-dependent}. By restricting the possibility to randomize over actions or by restricting the amount of information from the history of a run that can affect the choice of a scheduler, we obtain the following types of schedulers:
A scheduler $\sched$ is called \emph{deterministic} if it does not make use of the possibility to randomize over actions, i.e., if $\sched(\fpath)$ is a Dirac distribution
for each path $\fpath$. 
Such a scheduler $\sched$ can be viewed as a function that assigns an action
to each finite path $\fpath$. 
A scheduler
$\sched$ is called \emph{memoryless} if $\sched(\fpath)=\sched(\fpath^\prime)$ for
all finite paths $\fpath$, $\fpath^\prime$ with $\last(\fpath)=\last(\fpath^\prime)$.
In this case, $\sched$ can be viewed as a function
that assigns to each state $s$ a distribution over $\Act(s)$.
A memoryless deterministic scheduler hence can be seen as a function from states to actions.
In an MDP with a weight function, a scheduler $\sched$ is said to be \emph{weight-based} if
$\sched(\fpath)=\sched(\fpath')$ for all finite paths $\fpath$, $\fpath'$
with $\wgt(\fpath)=\wgt(\fpath')$ and $\last(\fpath)=\last(\fpath')$.
Such a scheduler assigns distributions over actions to state-weight pairs from $S\times \mathbb{Z}$.

\paragraph{Probability measure.} Given an MDP $\cM=(S,\Act,P,\sinit)$ and a scheduler $\sched$, we obtain a probability measure $\Pr^{\sched}_{\cM,s}$ on the set of maximal paths of $\cM$ that start in $s$:
For each finite path $\pi = s_0 \, \act_0 \, s_1 \, \act_1 \,  \ldots \act_{k-1} \, s_k$ with $s_0=s$, we denote the cylinder set of all its maximal extensions by $\Cyl(\pi)$. The probability mass of this cylinder set is then given by
\[
\Pr^{\sched}_{\cM,s}(\Cyl(\pi)) = P(\pi) \cdot \Pi_{i=0}^{k-1} \sched(s_0 \, \dots \, s_i) (\alpha_i).
\]
Recall that $\sched(s_0 \, \dots \, s_i) $ is a probability distribution over actions and that $\sched(s_0 \, \dots \, s_i) (\alpha_i)$ denotes the probability that the scheduler $\sched$ chooses action $\alpha$ after the prefix $s_0 \, \dots \, s_i$ of $\pi$. 
The set of cylinder sets forms the basis of the standard tree topology on the set of maximal paths. By Carath\'eodory's extension theorem, we can  extend the pre-measure $\Pr^{\sched}_{\cM,s}(\Cyl(\pi))$ defined on the cylinder sets to a probability measure on the Borel $\sigma$-algebra of the space of  maximal paths with the standard tree topology.
We sometimes drop the subscript $s$ if $s$ is the initial state $\sinit$ of~$\cM$. In a Markov chain $\cN$, we drop the reference to a scheduler and write $\Pr_{\cN,s}$.

Let $X$ be a random variable on the set of maximal paths of $\cM$ starting in $s$, i.e., $X$ is a function assigning values from $\mathbb{R}\cup\{-\infty,+\infty\}$ to maximal paths. We denote the expected value of $X$ under the probability measure $\Pr^{\sched}_{\cM,s}$ by $\mathbb{E}_{\cM,s}^{\sched} (X)$.

The values we are typically interested in are the worst- or best-case  probabilities of an event or the worst- or best-case expected values of a random variable. Worst or best case refers to the possible ways to resolve the non-deterministic choices. Hence, these values are formally expressed by taking the supremum or infimum over all schedulers. Given an MDP $\cM$, a state $s$, and an event, i.e., a  set of maximal paths, $E$, or a random variable $X$ on the maximal paths of~$\cM$, we define
\begin{align*}
\Pr^{\max}_{\cM,s} (E) &=\sup_\sched \Pr^{\sched}_{\cM,s} (E),  &&&
\Pr^{\min}_{\cM,s} (E) &=\inf_\sched \Pr^{\sched}_{\cM,s} (E), \\
\mathbb{E}^{\max}_{\cM,s} (X) &=\sup_\sched \mathbb{E}^{\sched}_{\cM,s} (X), \text{ and}&&&
\mathbb{E}^{\min}_{\cM,s} (X) &=\inf_\sched \mathbb{E}^{\sched}_{\cM,s} (X),
\end{align*}
 where $\inf$ and $\sup$ range over all schedulers $\sched$ for $\cM$.
 
 We use LTL-like notation  such as ``$\lozenge$(accumulated weight $<0$)'' to denote the event that a prefix of a path has a negative accumulated weight. Note that this event expresses the termination of a one-counter MDP in our view of one-counter MDPs as MDPs with a weight-function taking only values in $\{-1,0,+1\}$.
 
 \paragraph{Classical stochastic shortest path problem.}
 Let $\cM$ be an MDP with
 a weight function $\wgt \colon S \times \Act \to \mathbb{Z}$ and
  a designated set of terminal goal states $\Goal$.
 We define the following random variable $\rawdiaplus \Goal$ on maximal paths $\zeta$ of $\cM$ as follows:
 \[
 \rawdiaplus \Goal (\zeta) = \begin{cases}
 \wgt(\zeta) & \text{ if }\zeta\vDash \Diamond \Goal,\\
 \mathit{undefined} & \text{ otherwise}. 
 \end{cases}
 \]
The expected accumulated weight before reaching $\Goal$ under a scheduler $\sched$ is given by the expected value $\mathbb{E}^{\sched}_{\cM,\sinit}(\rawdiaplus\Goal)$. It is evident that this expected value is only defined if $\Pr^{\sched}_{\cM,\sinit}(\rawdiaplus\Goal)=1$. The \emph{classical stochastic shortest path problem} asks for the optimal value 
\[
\mathbb{E}^{\max}_{\cM,\sinit}(\rawdiaplus\Goal)=\sup_\sched\mathbb{E}^{\sched}_{\cM,\sinit}(\rawdiaplus\Goal)
\]
where the supremum ranges over all schedulers $\sched$ with $\Pr^{\sched}_{\cM,\sinit}(\rawdiaplus\Goal)=1$.
The classical stochastic shortest path problem can be solved in polynomial time \cite{bertsekas1991,deAlfaro1999,lics2018}.

\section{Outline of the Positivity-hardness proofs}
\label{sec:outline}

The Positivity-hardness results in this paper are obtained by  sequences of reductions depicted in Figure \ref{fig:overview_positivity}. The key steps for these sequences  are the three direct reductions from the Positivity-problem to the threshold problems for the maximal termination probability of one-counter MDPs, the maximal partial expectation, and the maximal conditional value-at-risk, respectively.

\subsection{Structure of the MDP constructed for the direct reductions from the Positivity problem}\label{sec:structure_MDP}

\begin{figure}[t] 
\begin{center}
  \includegraphics[width=0.4\linewidth]{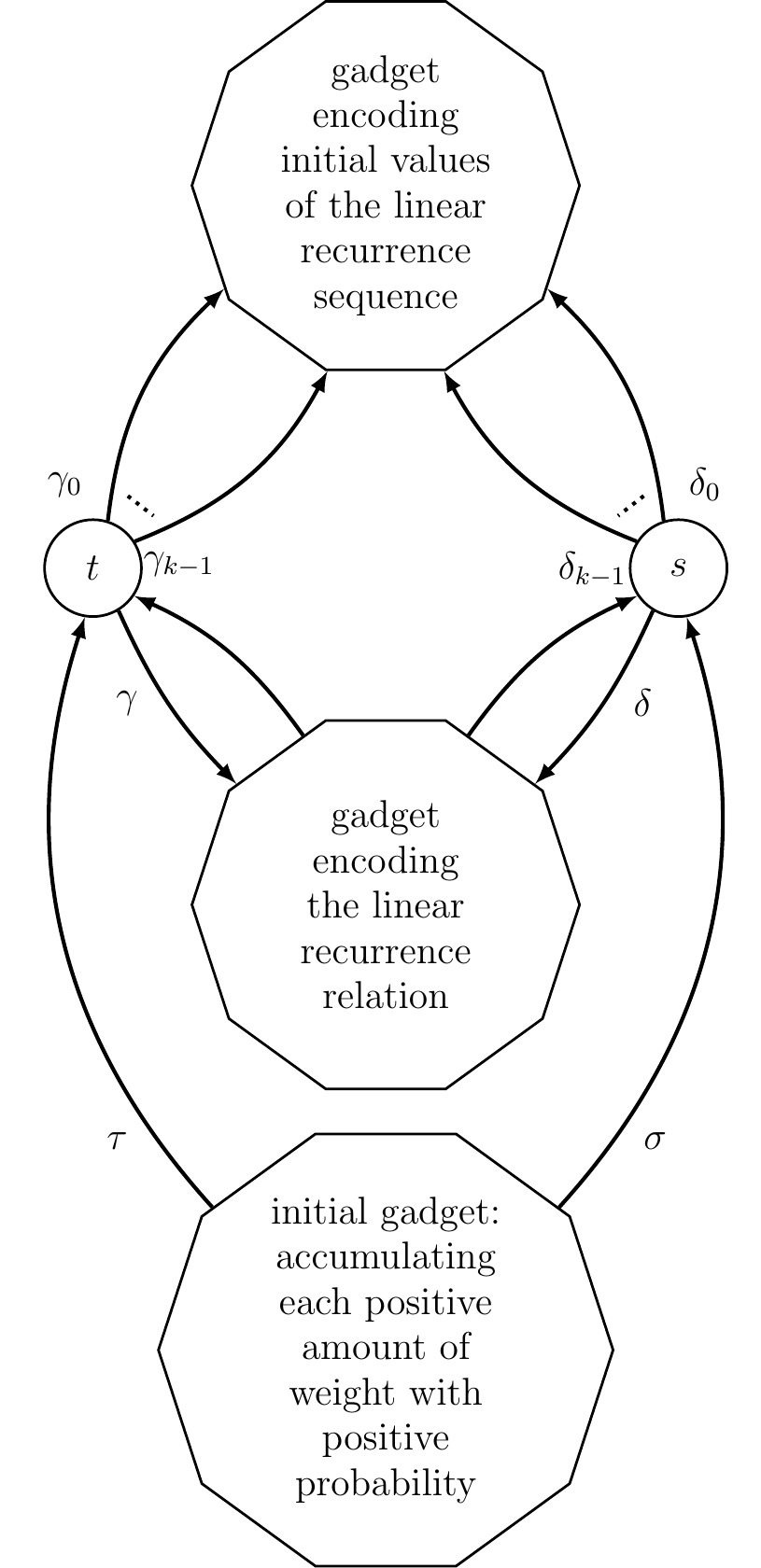}
\end{center}
\caption{Interplay between the MDP-gadgets.}
\label{fig:structure_gadgets}
\end{figure}

The three direct reductions from the Positivity problem (at the top of Figure \ref{fig:overview_positivity}) follow a modular approach: The MDPs constructed for the reductions are obtained by putting together three gadgets as sketched in  Figure~\ref{fig:structure_gadgets}. One gadget encodes a linear recurrence relation exploiting the dependency of optimal values from different starting states after different amounts of weight have been accumulated in the history of a run onto each other. A second gadget encodes the initial values of a linear recurrence sequence. Together, these two gadget allow us to encode  linear recurrence sequences. Finally, an initial gadget is added in which each positive amount of weight $w$ is accumulated with positive probability. Afterwards, the gadget is left and a scheduler has to 
decide how to leave the initial gadget. The optimal decision if weight~$w$ has been accumulated directly corresponds to whether the $w$th member of the  given linear recurrence sequence is non-negative.

More precisely, let a rational linear recurrence sequence be given in terms of the initial values $u_0,\dots, u_{k-1}$ and the coefficients $\alpha_1,\dots,\alpha_k$ of the linear recurrence relation.
The three gadgets are connected via two states $s$ and $t$ as depicted in Figure \ref{fig:structure_gadgets}. 
In state $t$ and $s$, actions $\gamma_0,\dots,\gamma_{k-1}$ and  $\delta_0,\dots,\delta_{k-1}$, respectively, leading to the gadget encoding the initial values and action $\gamma$ and $\delta$, respectively leading to the gadget encoding the linear recurrence relation are enabled. The gadgets will be constructed such that an optimal scheduler has to choose action~$\gamma_i$ or $\delta_i$ if the accumulated weight in state $t$ or $s$ is a value $i$ with $0\leq i < k$ and that it has to choose action $\gamma$ if the accumulated weight is at least $k$.
After $\gamma$ or $\delta$ is chosen, the accumulated weight is decreased within the gadget encoding the linear recurrence relation before the MDP moves back to the states $s$ and $t$ with positive probability.

Let us now denote the maximal possible value for the quantity of interest when starting in one of the states $t$ and $s$ with accumulated weight $w$ by $V(t,w)$ and $V(s,w)$. The linear recurrence relation will be found in the difference $d(w)\eqdef V(t,w)-V(s,w)$. If the accumulated weight is $0\leq i <k$, the gadget encoding the initial values will make sure that $d(i)=V(t,i)-V(s,i)=u_i$. For each of the three direct reductions from the Positivity problem, we construct one such gadget tailored to the three respective quantities.

For accumulated weights $w$ of at least $k$, the gadget encoding the recurrence will exploit the dependency of the optimal values $V(t,w)$ and $V(s,w)$ on the optimal values when starting with lower accumulated weight. This gadget can be used in all reductions and will be described in the next subsection.

Put together, these two gadgets ensure that $d(w)=u_w$ for all $w\geq 0$. To complete the reductions, we add an initial gadget $\cI$ depicted in Figure \ref{fig:gadget_initial} in which each positive amount of weight $w$ is accumulated with positive probability. Afterwards, a scheduler has to choose whether to move to state $t$ or state $s$ via the actions $\tau$ and $\sigma$, respectively. It is optimal to move to $t$ if and only if $u_w\geq 0$. Let now $\sched$ be the scheduler  always choosing $\tau$ in the initial gadget and afterwards behaving optimally when choosing from $\gamma_0,\dots, \gamma_{k-1}$ and $\gamma$ or $\delta_0,\dots, \delta_{k-1}$ and $\delta$ as described above. This scheduler is optimal if and only if the given linear recurrence sequence is non-negative.
The final step to complete the reduction is to compute the value $V^{\sched}(\sinit,0)$ that is achieved by $\sched$ starting from the initial state. In all three reductions, we can compute this rational value via converging matrix series. The optimal value $V^{\max}(\sinit,0)$ that can be achieved from the initial state now satisfies 
\[
V^{\max}(\sinit,0) \leq V^{\sched}(\sinit,0)
\]
if and only if the given linear recurrence sequence is non-negative.

\begin{figure}[t]
  \begin{center}
    \includegraphics[width=0.35\textwidth]{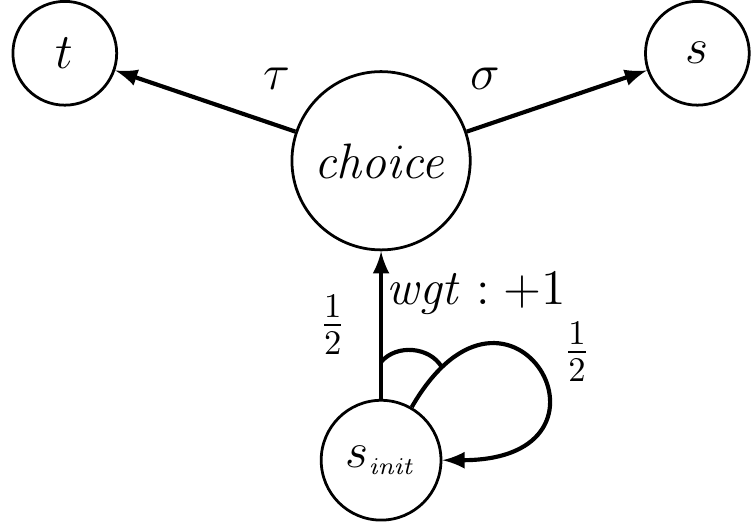}
   \end{center}

 \caption{The initial gadget $  \cI$.}\label{fig:gadget_initial}
\end{figure}

\subsection{MDP-gadget for linear recurrence relations}\label{sec:gadget_recurrence}

In this section, we demonstrate how to construct the gadget ensuring that the difference of optimal values $V(t,w)-V(s,w)$ follows a given linear recurrence \emph{relation} with respect to different weight levels $w$. 
In the next section, the initial values of a linear recurrence sequence will be encoded in MDP-gadgets tailored to the different  quantities we address.

\paragraph{Optimality equations.}
Let us start by the following observations on the well-known relation between the optimal values at different states in the classical stochastic shortest path  problem, i.e., the maximal expected accumulated weights before reaching a goal state (defined in Section~\ref{sec:prelim}).
Let $\cM=(S,\Act,P,\sinit,\wgt, \Goal)$ be an MDP.
The solution to the classical stochastic shortest path problem satisfies the so-called \emph{Bellman equation}. If $V(s)$ denotes the value when starting in state $s$, i.e., the maximal expected accumulated weight before reaching $\Goal$ from state $s$, then
\[
V(s)=\max_{\alpha\in \Act(s)} \wgt(s,\alpha)+\sum_{t\in S} P(s,\alpha,t)\cdot V(t)
\]
for $s\not\in \Goal$ and $V(s)=0$ for $s\in \Goal$. This simple form of optimality equation implies   the existence of optimal memoryless deterministic schedulers for the classical stochastic shortest path problem
(in case optimal schedulers exist, i.e., if the optimal values are finite).

For problems like the optimization of the termination probability of one-counter MDPs, it is, however, clearly not sufficient to consider the optimal values only in dependency of  the starting state. The  counter-value, i.e. the weight that has been accumulated so far, is essential.
 So, let $V(s,w)$ denote the maximal termination probability of a one-counter MDP when starting in state $s$ with counter value $w$. Letting $V(s,w)=1$ if $w<0$, we obtain the following equation for all states $s$ and all values $w\geq 0$:
\[
V(s,w)=\max_{\alpha\in \Act(s)} \sum_{t\in S} P(s,\alpha,t)\cdot V(t,w+\wgt(s,\alpha)).\tag{$\ast$}
\]
Already in this equation, the value $V(s,w)$ hence possibly depends on  values of the form $V(s,w-i)$ for some $i$. 
We want to exploit this interrelation to encode linear recurrence relations
\[
u_{n+k} = \alpha_1 u_{n+k-1} + \dots + \alpha_k u_n 
\]
 into the optimal values $V(s,w)$.
Of course, the values $P(s,\alpha,t)$ are all non-negative. So, we cannot directly encode a linear recurrence into the optimal values for different weight levels at one state as the coefficients might be negative.
To overcome this problem, we instead consider the difference $V(t,w)-V(s,w)$ for two different states $s$ and $t$.

\paragraph{Scaling down coefficients of a linear recurrence sequence.}
Given the coefficients $\alpha_1,\dots, \alpha_k$, and initial values $u_0=\beta_0$, \dots, $u_{k-1}=\beta_{k-1}$ of a linear recurrence sequence, we have to assume that these are all sufficiently small for the following constructions.
So, let us  clarify why we can  assume this without loss of generality and let us provide precise bounds.
Let  $(u_n)_{n\geq 0}$ be a linear recurrence sequence specified by the initial values $u_0=\beta_0$, \dots, $u_{k-1}=\beta_{k-1}$ and the linear recurrence relation
$ u_{n+k} = \alpha_1 u_{n+k-1} + \dots + \alpha_k u_n $
for all $n\geq 0$. 
For any $\mu> 0$ and $\lambda> 0$, the sequence $(v_n)_{n\geq0}$ defined by $v_n=\mu \cdot \lambda^n \cdot u_n$ for all $n$
is non-negative if and only if $(u_n)_{n\geq 0}$ is non-negative.  Furthermore, it satisfies $v_i = \mu \cdot \lambda^i \cdot \beta_i$ for $i<k$ and
\[
v_{n+k} = \lambda\cdot \alpha_1\cdot  v_{n+k-1} + \lambda^2\cdot \alpha_2 \cdot v_{n+k-2}  +\dots+ \lambda^k \cdot \alpha_k \cdot v_{n}.
\]
By choosing $\lambda$ and $\mu$ appropriately, we can scale down the initial values and coefficients of the recurrence relation for any given input.

To obtain precise bounds that will be used throughout the following sections,  let $\alpha \eqdef \sum_{i=1}^k |\alpha_i|$.
and let  $\lambda \eqdef \min\left( \frac{1}{\alpha \cdot (5k+5)},  \frac{1}{  (5k+5)} \right)$. So, if $\alpha>1$, then $\lambda =  \frac{1}{\alpha \cdot (5k+5)}$ and else $\lambda =  \frac{1}{  (5k+5)}$.
 The value $\lambda$ can be computed in polynomial time.
As the numerical value of $k$ is linear in the size of the given original input, the coefficients $\alpha_1^\prime\eqdef \lambda\cdot \alpha_1, \alpha_2^\prime \eqdef  \lambda^2\cdot \alpha_2, \dots, \alpha_k^\prime \eqdef \lambda^k \cdot \alpha_k $ of the linear recurrence of the sequence 
$(v_n)_{n\geq0}$ can be computed in polynomial time as well. The choice of~$\lambda$ ensures that $\sum_{i=1}^k|\alpha_i^\prime| < \frac{1}{5k+5}$.

Let now $\alpha^\prime \eqdef \sum_{i=1}^k |\alpha_i^\prime|$ and $\beta \eqdef \max_{0\leq j < k} |\beta_j|$.  
We can choose $\mu \eqdef \frac{\min(\alpha^\prime,1)}{4k^{2k+2}\cdot \beta}$. Again, since the value $k$ is linear in the size of the original input, $\mu$ can be computed in polynomial time.
The initial values of the new sequence $(v_n)_{n\geq0}$ are now  $\beta_i^\prime \eqdef v_i = \mu \cdot \lambda^i \cdot \beta_i$ for $i<k$, computable in polynomial time.
The choice of $\mu$ guarantees that $\max_{0\leq j <k} \beta_j^\prime< \min(\frac{1}{4k^{2k+2}},\frac{\alpha^\prime}{4})$.

Since this transformation can be carried out in polynomial time, we can w.l.o.g. from now on work under the following assumption:
\begin{assumption}
\label{ass:1}
Given the coefficients $\alpha_1,\dots, \alpha_k$, and initial values $u_0=\beta_0$, \dots, $u_{k-1}=\beta_{k-1}$ of a linear recurrence sequence, we assume that 
\[
\alpha\eqdef \sum_{i=1}^k|\alpha_i| < \frac{1}{5k+5}\text{ and that }\max_{0\leq j <k} \beta_j< \min(\frac{1}{4k^{2k+2}},\frac{\alpha}{4}).
\]
\end{assumption}

\paragraph{MDP-gadget for linear recurrence relations.}
Given the coefficients $\alpha_1,\dots, \alpha_k$ of a linear recurrence relation satisfying Assumption \ref{ass:1},
we construct  the  MDP-gadget  depicted in Figure \ref{fig:gadget}.
The gadget  contains states  $s$,  $t$, and $\trap$ as well as $s_1,\dots, s_k$ and $t_1,\dots, t_k$. In state $t$, an action~$\gamma$ is enabled which has weight $0$ and leads to state $t_i$ with probability $\alpha_i$ if $\alpha_i>0$ and to state $s_i$ with probability $|\alpha_i|$ if $\alpha_i<0$ for all $i$. The remaining probability leads to $\trap$. From each state~$t_i$, there is an action leading to $t$ with weight $-i$. The action $\delta$ enabled in $s$ as well as the actions leading from states $s_i$ to $s$ are constructed  analogously. If $\alpha_i$ is negative, action $\delta$ reaches state $t_i$ with probability $|\alpha_i|$. Otherwise it reaches $s_i$ with probability $\alpha_i$.
The state $\trap$ is absorbing.
As the gadget depends on the inputs $\bar{\alpha}=(\alpha_1,\dots,\alpha_k)$, we call it $\cG_{\bar{\alpha}}$.

\begin{figure}[t]
  \includegraphics[width=0.7\linewidth]{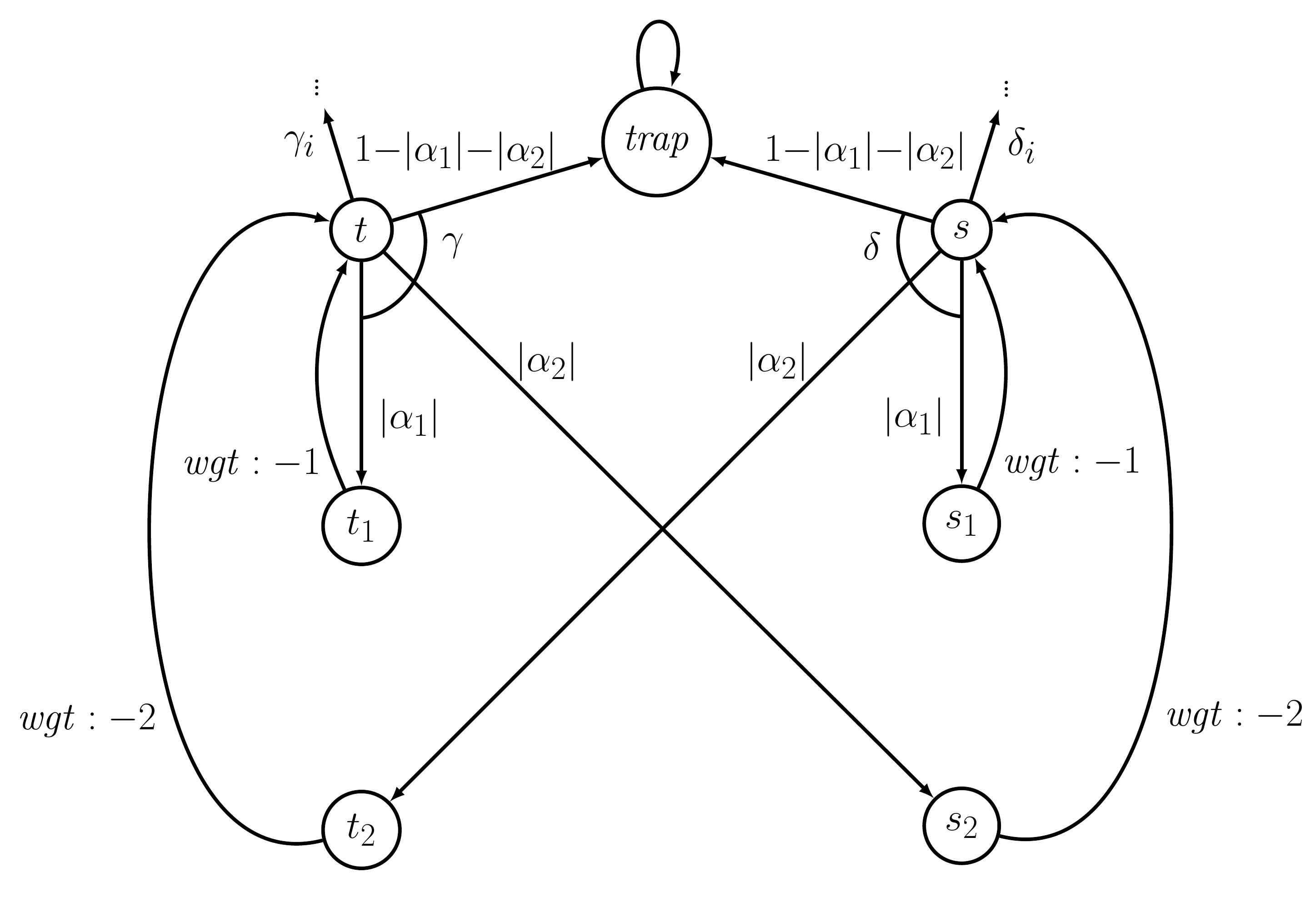}
     \centering

        \caption[The gadget to encode  linear recurrence relations.]{The gadget $\cG_{\bar{\alpha}}$ to encode  linear recurrence relations. The example here is depicted for a linear recurrence of depth $2$ with $\alpha_1\geq 0$ and $\alpha_2< 0$. The outgoing actions $\gamma_i$ and $\delta_i$ lead to the gadget encoding initial values as depicted in Figure \ref{fig:structure_gadgets}. }
        \label{fig:gadget}
\end{figure}

This gadget $\cG_{\bar{\alpha}}$ will be integrated into  MDPs without further outgoing edges from states $s_1,\dots,s_k,t_1,\dots,t_k$.
For any optimization problem for which the optimal values $V$  depend on the state and the weight accumulated so far and satisfy equation ($\ast$), we can encode a linear recurrence in an MDP containing this gadget (and possibly further actions for state $t$ and $s$):
If we know that an optimal scheduler chooses action $\gamma$ in state $t$ and action $\delta$ in state $s$ if the accumulated weight is $w$, then
\allowdisplaybreaks[0]
\begin{align*}
V(t,w)-V(s,w) 
=&\left(1-\sum_{i=1}^k |\alpha_i| \right)\left( V(\goal,w)-V(\goal,w) \right)\,\,\, + \\
&\,\,\, \sum_{1\leq i \leq k, \, \alpha_i\geq 0} \alpha_i V(t,w{-}i) - \alpha_i V(s,w{-}i)\,\,\,\,\, +  \\
& \,\,\, \sum_{1\leq i \leq k,\, \alpha_i< 0} (-\alpha_i )V(s,w{-}i)+ (-\alpha_i) V(t,w{-}i) \\
=& \sum_{i=1}^k \alpha_i \cdot (V(t,w{-}i) -V(s,w{-}i)). 
\end{align*}

Note that this linear recurrence relation also holds for the optimal values in the classical stochastic shortest path problem for example. So, the gadget alone is not yet enough for a hardness proof. 
The missing ingredient is the encoding of the initial values of a linear recurrence sequence. 
In order to include the encoding of the initial values in our approach, it is necessary that optimal schedulers cannot be chosen to be memoryless.
The optimal decisions have to depend on the weight that has been accumulated in the history of a run. If this is the case, we aim to encode the initial values by  adding further outgoing actions to the states $t$ and $s$. By fine-tuning the weights and probabilities of these actions, we can achieve that for small weights $w$ some of the new actions are optimal while for large weights the actions $\gamma$ and $\delta$ of the gadget are optimal. If we manage to design the other actions such that  the differences $V(t,w+i)-V(s,w+i)$ are equal to given starting  values $\beta_i$  for a sequence of weights $w,w+1,\dots, w+k-1$ while  actions~$\gamma$ and $\delta$ are optimal for weights of at least $w+k$, we can encode arbitrary linear recurrence sequences.
This is the goal of the subsequent section.


\section{Reductions from the Positivity problem}
\label{sec:reduction}

To encode initial values of a linear recurrence sequence, we construct further MDP gadgets. 
For  the termination  probability and expected termination time of one-counter MDPs and for partial expectations, we can construct these gadgets directly.
For the conditional value-at-risk, we use an intermediate auxiliary random variable. 
Putting together these gadgets with the gadget $\cG_{\bar{\alpha}}$ from the previous section, we obtain the basis for the Positivity-hardness results of the respective threshold problems.
The Positivity-hardness of the remaining problems is  obtained as a consequence of these results via further reductions. An overview of the chains of reductions used is presented in Figure \ref{fig:overview_positivity}.

\subsection{One-counter MDPs, energy objectives, cost problems, and quantiles}\label{sec:positivity_oc}

The first problem we will show to be Positivity-hard is the threshold problem for the optimal termination probability of one-counter MDPs. From this result, Positivity-hardness results for energy objectives, cost problems, and the computation of quantiles follow easily. Afterwards, we adjust the reduction to show Positivity-hardness of the threshold problem for the optimal expected termination time of almost-surely terminating one-counter MDPs.

\paragraph{Termination probability of one-counter MDPs.}
We formulated the termination of a one-counter MDP in terms of weighted MDPs $\cM$. Recall that a one-counter MDP terminates if the counter value drops below zero. If we consider the weight that has been accumulated instead of the counter value, the quantities we are interested are $\Pr^{\mathrm{opt}}_{\cM}(\lozenge \text{ accumulated weight}<0)$ for $\mathrm{opt}=\max$ and $\mathrm{opt}=\min$.
 The main result we prove in  this section is the following:

\begin{theorem}\label{thm:positivity_oc-mdp}
The Positivity problem is reducible in polynomial time to the following problems:
Given an MDP $\cM$ and a rational $\vartheta\in(0,1)$,
 \begin{enumerate}
 \item
 decide whether $\Pr^{\max}_{\cM,\sinit}(  \lozenge (\text{accumulated weight $< 0$}))> \vartheta$.
\item decide whether
$\Pr^{\min}_{\cM,\sinit}(  \lozenge (\text{accumulated weight $< 0$}))< \vartheta$.
\end{enumerate}
\end{theorem}

Note that if weights are encoded in unary, we can transform a weighted MDP to a one-counter MDP that can only increase or decrease the counter value by $1$ in each step in polynomial time.
The MDPs that are constructed from a linear recurrence sequence of depth $k$ in the proof of Theorem \ref{thm:positivity_oc-mdp} will contain only weights with an absolute value of at most $k$.
So, they can be transformed to one-counter MDPs in time linear in the size of the original input and we conclude that the following two threshold problems for the optimal termination probability of one-counter MDPs are Positivity-hard:

\begin{corollary}
\label{cor:one-counter}
The Positivity problem is reducible in polynomial time to the following problems:
Given a one-counter MDP $\cM$ viewed as an MDP with weights in $\{-1,0,+1\}$ and a rational $\vartheta\in(0,1)$,
 \begin{enumerate}
 \item
 decide whether 
$\Pr^{\max}_{\cM,\sinit}(  \lozenge (\text{accumulated weight $< 0$}))> \vartheta$.
\item decide whether
$\Pr^{\min}_{\cM,\sinit}(  \lozenge (\text{accumulated weight $< 0$}))< \vartheta$.
\end{enumerate}
\end{corollary}

Among the direct reductions from the Positivity problem  we present, the construction of the gadget encoding the initial values of a linear recurrence sequence is arguably the simplest for these optimal termination probabilities.
In the formulation with weighted MDPs, the termination of a one-counter MDP  is moreover the complement of the energy objective ``$\square \text{ accumulated weight}\geq 0$''.
We will first prove Positivity-hardness for the threshold problem for  maximal termination probabilities and outline the necessary adjustments to show Positivity-hardness also for the threshold problem for minimal termination probabilities afterwards.

We split the proof of Theorem \ref{thm:positivity_oc-mdp} into four parts. First, we provide the construction of an MDP from a linear recurrence sequence. Then, we show that the linear recurrence sequence is correctly encoded in this MDP in terms of the maximal termination probabilities. To complete the proof of item 1, we then show how to compute the threshold $\vartheta$ for the threshold problem and how this establishes the correctness of the reduction. Finally, we show how to adapt the construction to prove hardness of the threshold problem for minimal termination probabilities.

\paragraph{Proof of Theorem \ref{thm:positivity_oc-mdp}(1): construction of the MDP.}
Given a linear recurrence sequence in terms of the rational coefficients $\alpha_1,\dots,\alpha_k$ of the linear recurrence relation as well as the rational initial values 
$\beta_0,\dots, \beta_{k-1}$ for $k\geq 2$, our first goal is to construct an MDP $\cM$ and a rational $\vartheta\in (0,1)$ such that 
\[\Pr^{\max}_{\cM,\sinit}( \lozenge (\text{ accumulated weight $< 0$}))> \vartheta \qquad\text{ if and only if }\qquad u_n< 0\text{ for some }n\geq 0.\]
By Assumption \ref{ass:1}, we can assume that the input values are sufficiently small. More precisely, we assume that
$\sum_{i=1}^k |\alpha_i| <1/(k+1)$ and that $0\leq \beta_j< 1/(k+1)$ for all $0\leq j \leq k-1$, which is ensured by the bounds in Assumption \ref{ass:1}, and because the Positivity problem becomes trivial if one of the values $\beta_j$ with $0\leq j \leq k-1$ is  negative.

We denote the supremum of possible termination probabilities in terms of the current state $s$ and counter value  (accumulated weight) $w$ by $p(s,w)$.
More precisely, in an MDP $\cM$ for $w\geq 0$, we define 
\[
p(s,w)\eqdef\Pr^{\max}_{\cM,s}(\lozenge \text{ accumulated weight }<-w).
\]
 The values $p(s,w)$ in an MDP with state space $S$ now satisfy the  optimality equation ($\ast$) from Section \ref{sec:gadget_recurrence} (where  $p(s,w)$ takes  the role of $V(s,w)$ in ($\ast$)), which we restate here for convenience.  We have $p(s,w)=1$ for all states $s$ and all $w<0$ and 
\[
p(s,w)=\max_{\alpha\in \Act(s)} \sum_{t\in S} P(s,\alpha,t)\cdot p(t,w+\wgt(s,\alpha)) \qquad \text{ for all $s\in S$ and $w\geq 0$.}
\]
So, to capture the linear recurrence relation, we will be able to make use of the gadget $\cG_{\bar{\alpha}}$ from Section \ref{sec:gadget_recurrence}.
The missing ingredient is  a gadget to encode the initial values of a linear recurrence sequence.

\begin{figure}[t]
  \begin{center}
    \includegraphics[width=0.6\linewidth]{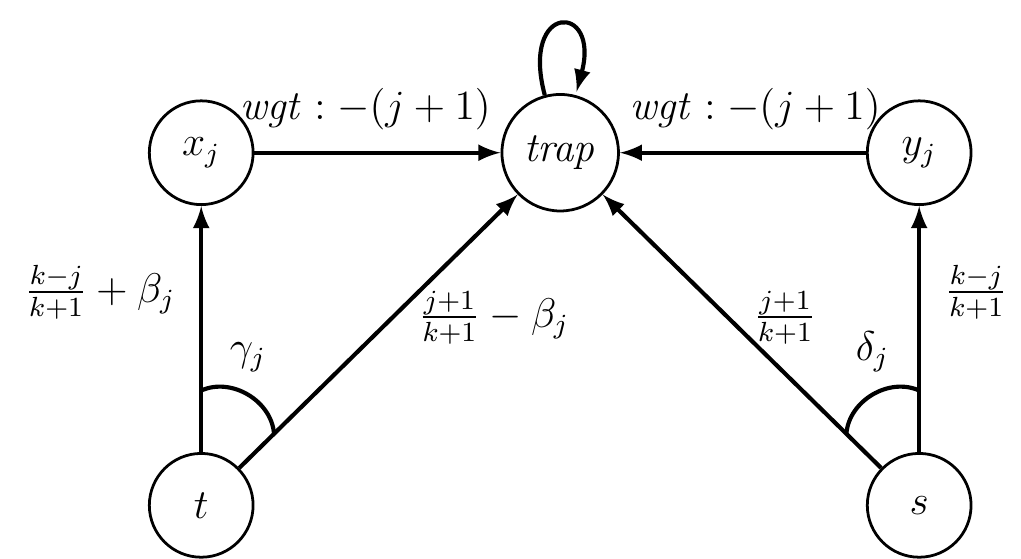}

   \end{center}

\caption{Gadget $\cO_{\bar{\beta}}$ encoding initial values of a linear recurrence sequence in terms of maximal termination probabilities of one-counter MDPs.}\label{fig:O_beta}
\end{figure}

The new gadget $\cO_{\bar{\beta}}$ encoding the initial values $\bar{\beta}$ is depicted in Figure \ref{fig:O_beta} and works as follows:
For $0\leq j \leq k-1$, the action $\gamma_j$ enabled in $t$ leads to state $x_j$ with probability $\frac{k-j}{k+1}+\beta_j$. By assumption on $\beta_j$, this probability is less than  $ \frac{k-j+1}{k+1}$. The remaining probability leads to $\trap$.
In state $s$, the action $\delta_j$ leads to $y_j$ with probability $\frac{k-j}{k+1}$ and to $\trap$ with the remaining probability.
For $0\leq j \leq k-1$, one reaches $\trap$ from $x_j$ and $y_j$ with probability $1$ and a counter change of~$-(j+1)$.

\begin{figure}[t]
  \begin{center}
    \includegraphics[width=\linewidth]{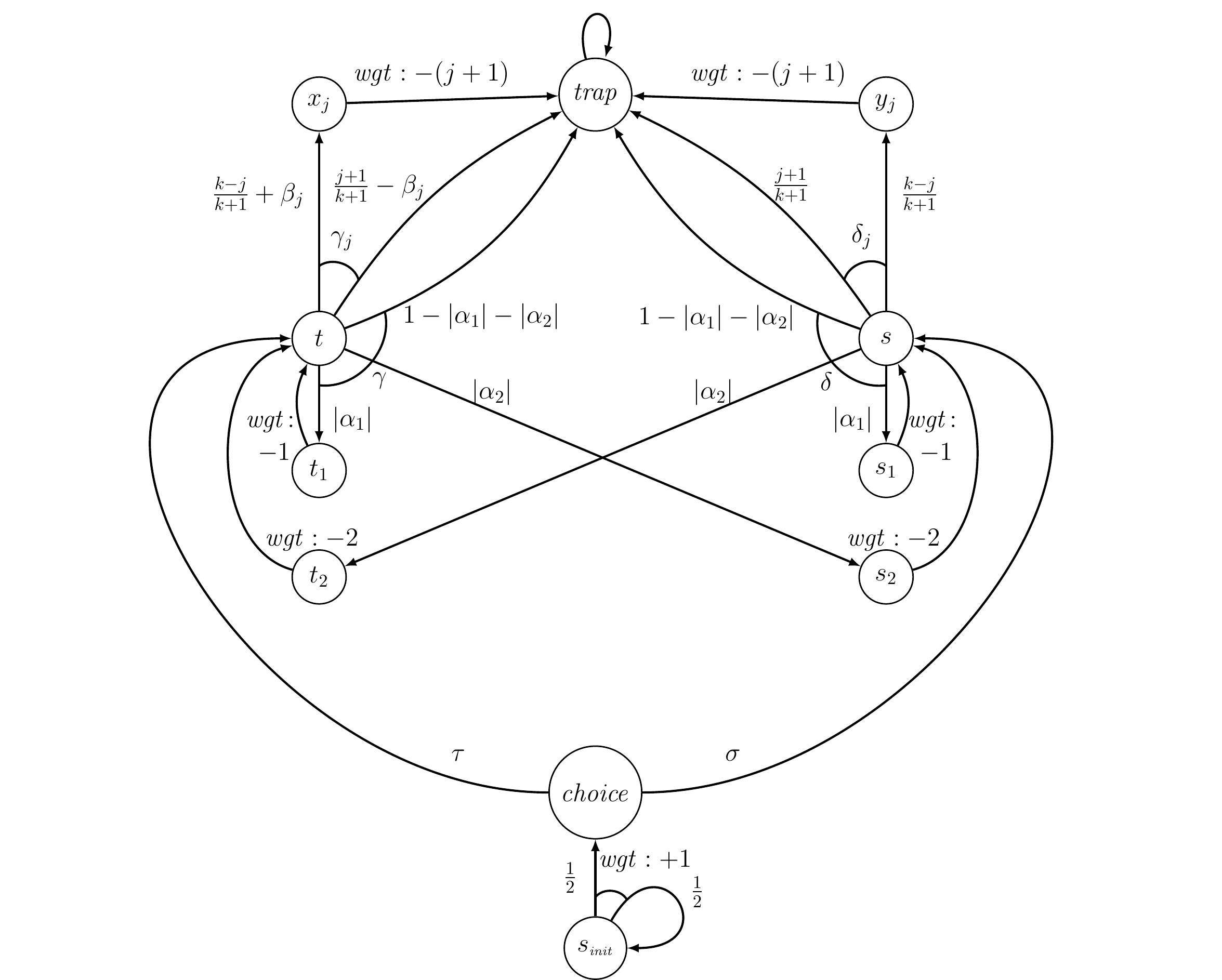}

   \end{center}

\caption[Full MDP for the reduction to the threshold problem for termination probabilities of one-counter MDPs. ]{Full MDP for the reduction to the threshold problem for termination probabilities of one-counter MDPs. The MDP contains the upper part for all $0\leq j \leq k-1$. The middle part is depicted for $k=2$, $\alpha_1\geq 0$, and $\alpha_2<0$.}\label{fig:MDP}
\end{figure}

Now, we glue together the initial gadget $\cI$ defined in Section \ref{sec:structure_MDP}, the gadget encoding the linear recurrence relation $\cG_{\bar{\alpha}}$ from Section \ref{sec:gadget_recurrence}, and the new gadget $\cO_{\bar{\beta}}$ at states $t$, $s$, and $\trap$.
The resulting MDP $\cM$ is depicted in Figure \ref{fig:MDP} -- for better readability, it is depicted for $k=2$ and assuming that $\alpha_1\geq0$ while $\alpha_2<0$.

\paragraph{Proof of Theorem \ref{thm:positivity_oc-mdp}(1): correctness of the encoding of the linear recurrence sequence.}
In this paragraph, we show that the initial linear recurrence sequence is indeed encoded in the maximal termination probabilities when starting from states $t$ and $s$ with different counter values, i.e., values of accumulated weight as described in Section \ref{sec:structure_MDP}.
More precisely, let $(u_n)_{n\geq 0}$ be the linear recurrence sequence given by the initial values $\beta_0,\dots \beta_{k-1}$ and the coefficients $\alpha_1,\dots, \alpha_k$ of the linear recurrence relation. We prove the following:

\begin{lemma}\label{lem:recurrence_OC}
For each $w\geq 0$, we have 
\[
p(t,w)-p(s,w)=u_w
\]
where $p(r,w)$ denotes the maximal termination probability from state $r\in \{s,t\}$ when starting with  accumulated weight $w$ as defined above. 
\end{lemma}

\begin{proof}
For the correct interplay of the gadgets $\cG_{\bar\alpha}$ and $\cO_{\bar\beta}$, the optimal decisions in states $t$ and $s$ for different values of accumulated weights, i.e., different counter-values, are crucial. 
In order to terminate, the accumulated weight has to drop below $0$  before reaching $\trap$. As soon as the trap state is reached with non-negative accumulated weight, the process cannot terminate anymore. The optimal decision in order to maximize the termination probability in state $t$ is now easy to determine. 
Let $\ell$ be the current weight. If $0\leq \ell \leq k-1$, choosing action $\gamma$ leads to termination with probability less than $1/(k+1)$ as $\trap$ is reached immediately with probability at least $k/(k+1)$ due to our assumption that $\sum_{i\leq k}|\alpha_i|<1/(k+1)$. Choosing action~$\gamma_j$ makes it impossible to terminate if $\ell>j$. If $\ell \leq j$, then choosing $\gamma_j$  lets the process terminate if~$x_j$ is reached. This happens with probability $\frac{k-j}{k+1}+\beta_j$. As $\beta_j<1/(k+1)$ for all $j$, the maximal termination probability is reached when choosing $\gamma_{\ell}$.
If $\ell\geq k$, then $\gamma_j$ leads to termination with probability $0$ for all $j$. Hence, action $\gamma$ is optimal. Analogously, we see that the optimal choice in state $s$ with weight $\ell$ is $\delta_{\ell}$ if $\ell\leq k-1$ and $\delta$ otherwise.

The linear recurrence sequence $(u_n)_{n\geq 0}$ now can be found in terms of the  difference 
\[
d(w)\eqdef p(t,w)-p(s,w).
\] For counter value $w\leq k-1$, we have seen that $\gamma_{w}$ and $\delta_{w}$, respectively, are the optimal actions. Hence, $d(w)=u_{w}$ in this case as we have just seen that the optimal termination probability when starting with weight $w\leq k-1$ is $\frac{k-w}{k+1}+\beta_w $ in $t$ and $\frac{k-w}{k+1}$ in $s$.
Furthermore, for $w>k-1$, actions $\gamma$ and $\delta$ are optimal. So by the construction of gadget $\cG_{\bar\alpha}$,
\begin{align*}
p(t,w)-p(s,w) 
=&\left(1-\sum_{i=1}^k |\alpha_i| \right)\left( p(\trap,w)-p(\trap,w) \right)\,\,\, + \\
&\,\,\, \sum_{1\leq i \leq k, \, \alpha_i\geq 0} \alpha_i p(t,w{-}i) - \alpha_i p(s,w{-}i)\,\,\,\,\, +  \\
& \,\,\, \sum_{1\leq i \leq k,\, \alpha_i< 0} (-\alpha_i )V(s,w{-}i)+ (-\alpha_i) p(t,w{-}i) \\
=& \sum_{i=1}^k \alpha_i \cdot (p(t,w{-}i) -p(s,w{-}i)).
\end{align*}
So, the  sequence of differences satisfies the  linear recurrence relation given by $\alpha_1,\dots, \alpha_k$. 
Therefore, $d(w)=u_{w}$ for all $w\geq 0$.
\end{proof}

\paragraph{Proof of Theorem \ref{thm:positivity_oc-mdp}(1): computation of the threshold $\vartheta$.}

The state $\mathit{choice}$ is reached with any positive accumulated weight with positive probability.
For the optimal choices in the state $\mathit{choice}$ with accumulated weight $w$, we observe that choosing $\tau$ is optimal if and only if $d(w)\geq 0$. By Lemma \ref{lem:recurrence_OC}, this holds if and only if $u_{w}\geq 0$. 

Consider now the scheduler $\sched$ which always chooses $\tau$ in state $\mathit{choice}$ and afterwards behaves according to the optimal choices as described in the proof of Lemma \ref{lem:recurrence_OC}.
This scheduler~$\sched$ is optimal if and only if the sequence $(u_n)_{n\geq 0}$ is non-negative.
To complete the reduction, we will compute the value 
\[\vartheta\eqdef\Pr_{\cM,\sinit}^{\sched}(\lozenge  (\text{accumulated weight }< 0)).\]
 We will see that $\vartheta$ is a rational computable in polynomial time and we know that 
\[
\Pr_{\cM,\sinit}^{\max}(\lozenge (\text{accumulated weight }< 0)) \leq \vartheta
\]
if and only if the scheduler $\sched$ is optimal which is the case if and only if $(u_n)_{n\geq 0}$ is non-negative.

\begin{lemma}\label{lem:compute_threshold_OC}
In the constructed MDP $\cM$,
the value $\vartheta= \Pr_{\cM,\sinit}^{\sched}(\lozenge  (\text{accumulated weight }< 0))$ can be computed in polynomial time.
\end{lemma}

\begin{proof}
In order to compute the value $\vartheta$, 
we first provide a recursive expression of the maximal termination probabilities $p(t,w)$ and $p(s,w)$. By the definition of $\sched$, these are precisely the termination probabilities under $\sched$ when starting from $t$ or $s$ with some positive accumulated weight $w\in \mathbb{N}$ because $\sched$ behaves optimally as soon as state $t$ or $s$ has been reached.

For this recursive expression, we consider the following Markov chain $\cC$ for $n\in \mathbb{N}$ that is also depicted in Figure \ref{fig:Markov_chain} -- for better readability, it is  depicted for the case $k=2$ there:
The Markov chain $\cC$ has $5k$ states named $t_{-k+1}$, \dots, $t_{+k}$, $s_{-k+1}$, \dots, $s_{+k}$, and $\goal_{+1}$, \dots, $\goal_{+k}$. 
States $t_{-k+1}$, \dots, $t_{0}$, $s_{-k+1}$, \dots, $s_{0}$, and $\goal_{+1}$, \dots, $\goal_{+k}$ are terminal.
 For $0<i,j\leq k$, there are transitions from $t_{+i}$ to $t_{+i-j}$ with probability $\alpha_j$ if $\alpha_j>0$, to $s_{+i-j}$ 
 with probability $|\alpha_j|$ if $\alpha_j<0$, and to $\goal_{+i}$ with probability $1-|\alpha_1|-\ldots-|\alpha_k|$. Transitions from $s_{+i}$ are defined analogously.

\begin{figure}[t]
  \begin{center}
    \includegraphics[width=0.45\linewidth]{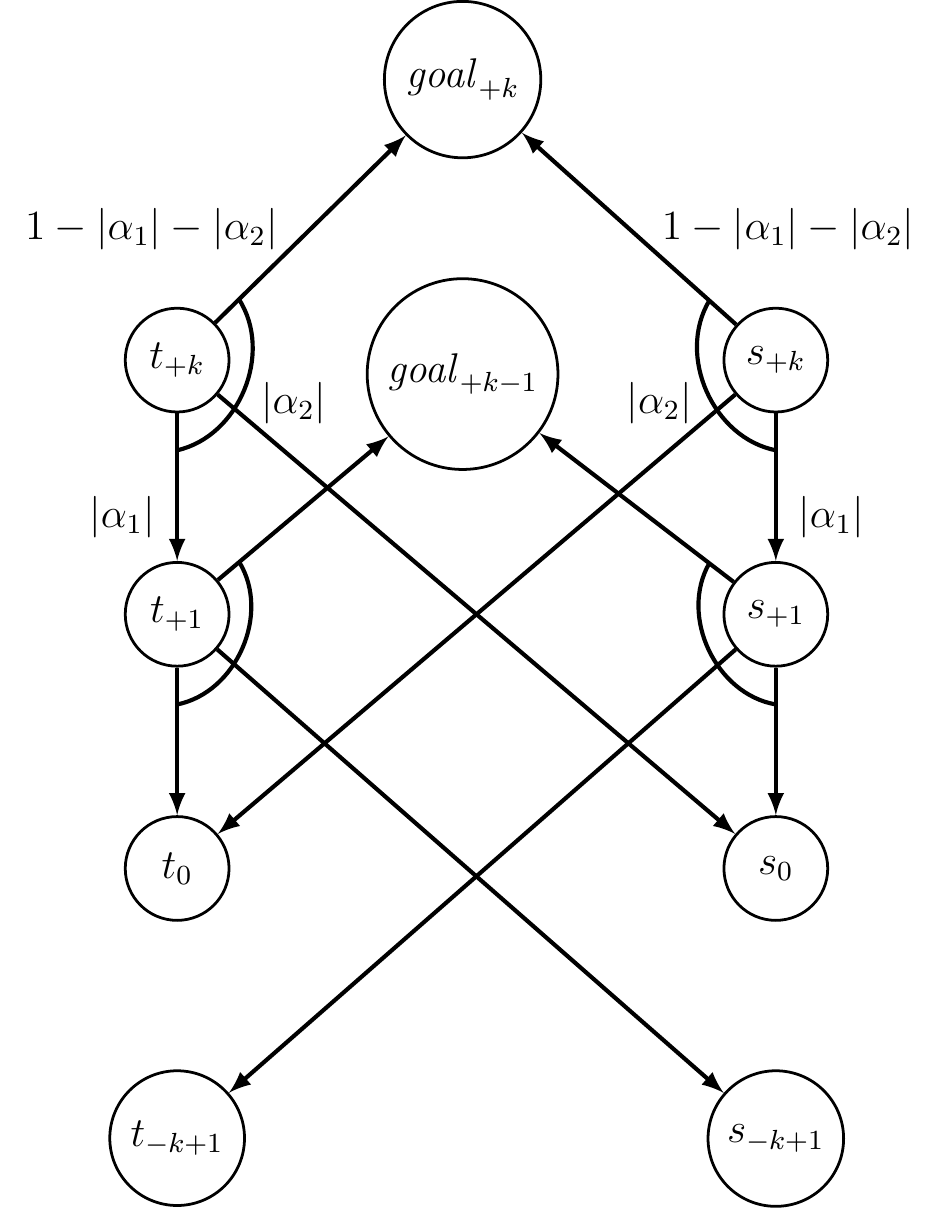}

\end{center}

\caption{The Markov chain $\cC$ depicted for $k=2$ with $\alpha_1\geq 0$ and $\alpha_2<0$.}\label{fig:Markov_chain}
\end{figure}
 
 The idea behind this Markov chain is that the reachability probabilities describe how, for arbitrary $n\in \mathbb{N}$ and $1\leq i \leq k$, 
 the values $p(t,nk+i)$ and $p(s,nk+i)$ depend on   the values $p(t,(n-1)k+j)$ and $p(s,(n-1)k+j)$ for $1\leq j \leq k$. 
 The transitions in $\cC$ behave as $\gamma$ and $\delta$ in $\cM$, but the decrease in the accumulated weight is explicitly encoded into the state space. 
 Namely, for  $n\in\mathbb{N}$ and $0<i\leq k$, we have
{ \begin{align}\label{eqn:cC1}
  p(t,nk+i) = \sum_{j=1}^k &\big(\Pr_{\cC, t_{+i}} (\lozenge t_{-k+j})\cdot p(t,(n{-}1)k+j) 
  +\Pr_{\cC, t_{+i}} (\lozenge s_{-k+j})\cdot p(s,(n{-}1)k+j) \big)  \tag{$\ast$}
\end{align}}
and analogously for $ p(s,nk+i)$.
We now group the optimal values together in the following column vectors
\[v_n = (p(t,nk+k-1) , p(t,nk+k-2), \ldots, p(t,nk), p(s,nk+k-1) , \ldots, p(s,nk))^\top \]
for $n\in \mathbb{N}$. In other words, this vector contains the optimal values for the partial expectation when starting in $t$ or $s$ with an accumulated weight from  $\{nk,\dots,nk+k-1\}$. 
The vector $v_0$ is the column vector
\[(p(t,k-1),\dots, p(t,0), \dots p(s,k-1),\dots,p(s,0))^\top
\]
and these values occur as transition probabilities in $\cM$ under the actions $\gamma_{k-1}, \dots, \gamma_{0}$ and $ \delta_{k-1},\dots, \delta_{0}$.

As the reachability probabilities in $\cC$ are rational and computable in polynomial time, we conclude from equation (\ref{eqn:cC1})
that there is a matrix $A\in \mathbb{Q}^{2k\times2k}$ computable in polynomial time such that 
$ v_{n+1} = A v_{n}$
for all $n\in\mathbb{N}$. 
So, $v_n=A^n v_0$ for all $n\in\mathbb{N}$.

As state $\mathit{choice}$ is reached with weight  $w$ with probability $(1/2)^w$ for all $w\geq 1$,
the value $\vartheta=
\sum_{w=1}^\infty (1/2)^w p(t,w)$.
Let $c=(\frac{1}{2^k} , \frac{1}{2^{k-1}} , \dots, \frac{1}{2^1} , 0 ,\dots, 0)$. 
Observe that for all $n\in\mathbb{N}$,
\[\left(\frac{1}{2^k}\right)^n \cdot c\cdot v_n= \sum_{i=1}^{k}  \frac{1}{2^{nk+i}} p(t,nk+i).\]

\noindent
Hence, we can write 
\begin{align*}
\vartheta &= \sum_{n=0}^\infty \left(\frac{1}{2^k}\right)^n \cdot c \cdot v_n - p(t,0)  = c \cdot \sum_{n=0}^\infty \left(\frac{1}{2^k}\right)^n \cdot v_n - p(t,0) \\
&=  c \cdot \sum_{n=0}^\infty \left(\frac{1}{2^k}\right)^n \cdot A^{n} \cdot v_{0} -p(t,0) = c \cdot \left(\sum_{n=0}^\infty \left(\frac{1}{2^k}\cdot A \right)^n \right) \cdot v_{0} -p(t,0).
\end{align*}
We have to subtract $p(t,0)$ as the state $\mathit{choice}$ cannot be reached with weight $0$, but the summand $1\cdot p(t,0)$ occurs in the sum. As $p(t,0)=\frac{k}{k+1}+\beta_0$, this does not cause a problem.

We claim that  the matrix series involved converges to a rational matrix. 
 We observe that the maximal row sum in $A$ is at most $|\alpha_1|{+}\ldots{+}|\alpha_k|<1$ because the rows of the matrix contain exactly the probabilities to reach $t_0$, \dots $t_{-k+1}$, $s_0$, \dots, and  $s_{-k+1}$ from a state $t_{+i}$ or $s_{+i}$ in $\cC$ for $1\leq i \leq k$. But the probability to reach $\goal_{+i}$ from these states is already $1{-}|\alpha_1|{-}\ldots{-}|\alpha_k|$. Hence, $\Vert A \Vert_{\infty}$, the operator norm induced by the maximum norm $\Vert \cdot \Vert_\infty$, which equals $\max_{i} \sum_{j=1}^{2k} |A_{ij}|$, is less than $1$.
So, in particular, also $\Vert \frac{1}{2^k} A \Vert_{\infty}<1$ and hence the  Neumann series $\sum_{n=0}^\infty \left(\frac{1}{2^k} A\right)^{n}$ converges to $\left(I_{2k}-\frac{1}{2^k} A\right)^{-1}$ where $I_{2k}$ is the identity matrix of size $2k{\times }2k$. So,
\[
\vartheta= c \cdot \left(I_{2k}-\frac{1}{2^k} A\right)^{-1} \cdot v_0 - p(t,0)
\]
is computable in polynomial time. 
\end{proof}

All in all, this finishes the proof of item (1) of  Theorem \ref{thm:positivity_oc-mdp}: 
We have seen that the MDP $\cM$ and the threshold $\vartheta$ can be constructed in  time polynomial in 
the size of the representations of $\alpha_1, \dots, \alpha_k$  and $\beta_0, \dots, \beta_{k-1}$. 
As
$\vartheta= \Pr_{\cM,\sinit}^{\sched}(\lozenge  (\text{accumulated weight }< 0))$,
we furthermore know that
\[\Pr_{\cM,\sinit}^{\max}(\lozenge ( \text{accumulated weight }< 0)) > \vartheta\]
if and only if the scheduler $\sched$ is not optimal. By Lemma \ref{lem:recurrence_OC}, this is the case if and only if
the given linear recurrence sequence $(u_n)_{n\geq 0}$ has a negative member: 
If $u_w<0$ for some $w\in \mathbb{N}$, then the following scheduler $\tsched$ achieves a value greater than $\vartheta$.
The scheduler $\tsched$ behaves like~$\sched$ except when in state $\mathit{choice}$ with accumulated weight $w$. In this case, $\tsched$ chooses $\sigma$ instead of $\tau$. As $\mathit{choice}$ is reached with accumulated weight $w$ with positive probability and 
$p(s,w)>p(t,w)$ this scheduler outperforms $\sched$ as it behaves optimally when reaching state $t$ with accumulated weight $w$ as shown in the proof of Lemma \ref{lem:recurrence_OC}.
If on the other hand $u_w\geq 0$ for all $w\in \mathbb{N}$, then $p(t,w)>p(s,w)$ for all $w$ and hence action $\tau$ is always optimal in state $\mathit{choice}$. As $\sched$ behaves optimally once $t$ or $s$ is reached and always chooses $\tau$, scheduler $\sched$ is indeed optimal in this case.

Finally,  we want to emphasize again that the absolute values of the weights in the constructed MDP are at most $k$. Hence, if we want to view $\cM$ as a one-counter MDP in which the counter value can only be increased or decreased by $1$ in each step, the constructed MDP becomes only polynomially larger after we replace the transitions with a weight $+w$ or $-w$ for a $1\leq w \leq k$ by a sequence of $w$ states decreasing or increasing the counter value, which allowed us to conclude Corollary \ref{cor:one-counter}.

\paragraph{Proof of Theorem \ref{thm:positivity_oc-mdp}(2).}

The construction we provided so far shows that the threshold problem for the \emph{maximal} termination probability of one-counter MDPs is Positivity-hard.
Using exactly the same ideas, we can show that the threshold problem for the \emph{minimal} termination probability is Positivity-hard as well. Let us describe the necessary changes in the construction that are also depicted in Figure \ref{fig:MDP_minimal_oc}.
We rename the state $\trap$ to $\trap^\prime$ and  add a transition with weight $-k$ to a new absorbing state $\trap$.
For all $0\leq j \leq k-1$,  now state $\trap$ is reached directly with probability $1$ and weight $-j$ from the states $x_j$ and $y_j$.
Furthermore, the probability to reach $x_j$ when choosing $\gamma_j$ in $t$ is changed to $\frac{j+1}{k+1}+\beta_j$ and the probability to reach $\trap^\prime$ is adjusted accordingly. The analogous change is performed for $\delta_j$.
Now, it is easy to check that the optimal choice to minimize the termination probability in state $t$ is to choose $\gamma$ if the accumulated weight is $\geq k$. In this case the probability of termination is less than $\frac{1}{k+1}$. If the accumulated weight is $0\leq \ell<k$, the optimal choice is $\gamma_\ell$. The analogous result holds in state $s$.  From then on the proof is analogous to the proof
for the maximal termination probability
 with the change that we have to consider the scheduler $\sched$ always choosing $\sigma$ in the state $\mathit{choice}$ this time. 
 This scheduler is optimal to minimize the termination probability if and only if the given linear recurrence sequence is non-negative.
  With these adjustments, we conclude:
\begin{corollary} \label{cor:positivity_oc}
The Positivity problem is reducible in polynomial time to the following problem:
Given an MDP $\cM$ and a rational $\vartheta\in(0,1)$, decide whether 
\[\Pr^{\min}_{\cM,\sinit}( \lozenge (\text{accumulated weight $< 0$}))< \vartheta.\]
\end{corollary}

 \begin{remark}\label{rem:strict}
 There is no obvious way to adjust the construction such that  the Positivity-hardness of the question whether  $\Pr^{\max}_{\cM,\sinit}( \lozenge (\text{accumulated weight $< 0$}))\geq \vartheta$ would follow. One attempt would be to provide an $\varepsilon$ such that  $\Pr^{\max}_{\cM,\sinit}( \lozenge (\text{accumulated weight $< 0$}))>\vartheta$ if and only if   $\Pr^{\min}_{\cM,\sinit}( \lozenge (\text{accumulated weight $< 0$})) \geq \vartheta + \varepsilon$. This, however, probably requires a bound on the position at which the given linear recurrence sequence first becomes negative. But this question lies at the core of the Positivity  problem. The analogous observation applies to
 the question  whether 
 $\Pr^{\min}_{\cM,\sinit}( \lozenge (\text{accumulated weight $< 0$}))\leq \vartheta$ and all Positivity-hardness results in the sequel.
 \end{remark}


\begin{figure}[t]
\begin{center}
  \includegraphics[width=0.65\linewidth]{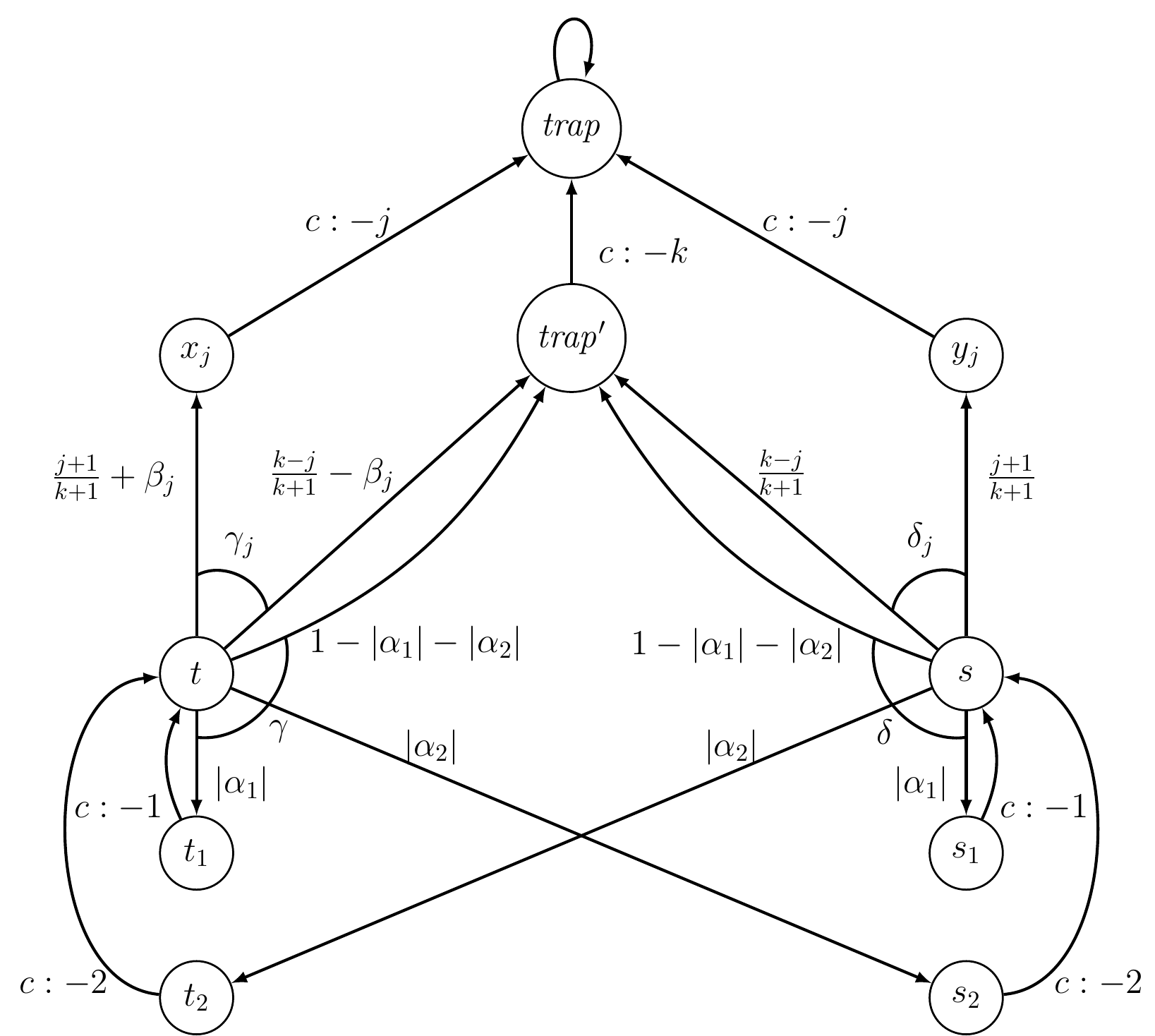}

   \end{center}

\caption[Necessary changes to the construction for the result for minimal termination probabilities.]{Necessary changes to the construction for the result for minimal termination probabilities. The initial component of the MDP is omitted here and stays unchanged.}\label{fig:MDP_minimal_oc}
\end{figure}

\paragraph{Energy objectives.}
As the energy objective $\Box(\text{accumulated weight $\geq 0$})$ is satisfied if and only if $\lozenge (\text{accumulated weight $< 0$})$ does not hold, the Positivity-hardness of the threshold problem of the optimal satisfaction probability of an energy objective follows easily. As
\[
\Pr^{\max}_{\cM,\sinit}( \Box (\text{accumulated weight $\geq 0$}))=1-\Pr^{\min}_{\cM,\sinit}( \lozenge (\text{accumulated weight $< 0$})),
\]
we conclude:
\begin{corollary}
The Positivity problem is reducible in polynomial time to the following problems:
Given an MDP $\cM$ and a rational $\vartheta\in(0,1)$, 
\begin{enumerate}
\item decide whether 
$\Pr^{\max}_{\cM,\sinit}( \Box (\text{accumulated weight $\geq 0$}))> \vartheta$.
\item decide whether 
$\Pr^{\min}_{\cM,\sinit}( \Box (\text{accumulated weight $\geq 0$}))< \vartheta.$
\end{enumerate}
\end{corollary}

\paragraph{Cost problems and quantiles.} The proof of the Positivity-hardness of the threshold problem for the termination probability of one-counter MDPs in fact also serves as a proof that cost problems and the computation of quantiles of the accumulated weight before reaching a goal state are Positivity-hard.
Observe that in the MDP constructed for Theorem \ref{thm:positivity_oc-mdp} and Corollary \ref{cor:positivity_oc}, almost all paths $\zeta$ under any scheduler satisfy 
$ \lozenge (\text{accumulated weight $< 0$})$ if and only if they satisfy $\rawdiaplus \trap(\zeta) < 0$ if and only if their total accumulated weight is less than $0$. Thus, we obtain the following corollary:
\begin{corollary}
The Positivity problem is reducible in polynomial time to the following problems:
Given an MDP $\cM$ with a designated set of trap states $\Goal$ and a rational $\vartheta\in(0,1)$, 
\begin{enumerate}
\item
decide whether 
$\Pr^{\max}_{\cM,\sinit}( \rawdiaplus \Goal < 0)> \vartheta$.
\item
 decide whether 
$\Pr^{\min}_{\cM,\sinit}(  \rawdiaplus \Goal < 0)< \vartheta$.
\end{enumerate}
\end{corollary}
The analogous result also holds for the total accumulated weight.

\paragraph{Termination times of one-counter MDPs.}
To conclude the section, we show that not only the threshold problems for optimal termination probabilities, but also for the optimal expected termination times in one-counter MDPs that terminate almost surely is Positivity-hard. 
We again work with weighted MDPs.
Let $T$ be the random variable that assigns to each path in a weighted MDP $\cM$ the length of the shortest prefix $\pi$ such that $\wgt(\pi)<0$.
To reflect precisely the behavior of a one-counter MDP, we now will work with  MDPs where the weight is reduced or increased by at most $1$ in each step. We make a small change to the MDP constructed for the proof of Corollary \ref{cor:positivity_oc} that is depicted in Figure \ref{fig:MDP_minimal_oc}. The initial component (that is not depicted) stays unchanged. For the remaining transitions, all transition reduce the weight or leave it unchanged. The transitions with weight $0$ do not occur directly after each other except for the loop at the state $\trap$ that we adjust in a moment. Hence, we can add additional auxiliary states such that along each path starting from $s$ or $t$ not reaching the state $\trap$, the weight is left unchanged and reduced by $1$ in an alternating fashion.
So, if a path starts in state $s$ or $t$ with accumulated weight $w$ and terminates (i.e., reaches accumulated weight $-1$) before reaching the state $\trap$ this takes $2(w+1)$ steps.
Now, we replace the loop at the state $\trap$ by the gadget depicted in Figure \ref{fig:MDP_minimal_oc_time} and let us call the resulting MDP $\cN$.
So, when reaching $\trap$ the accumulated weight is increased by $1$ before it is reduced in every other step until termination.
That means that if a path starting in state $s$ or $t$ with weight $w$ does not terminate before reaching $\trap$, the termination time is $2(w+1)+3$ steps.

\begin{sidefigure}[t]
  \begin{center}
    \includegraphics[width=0.25\textwidth]{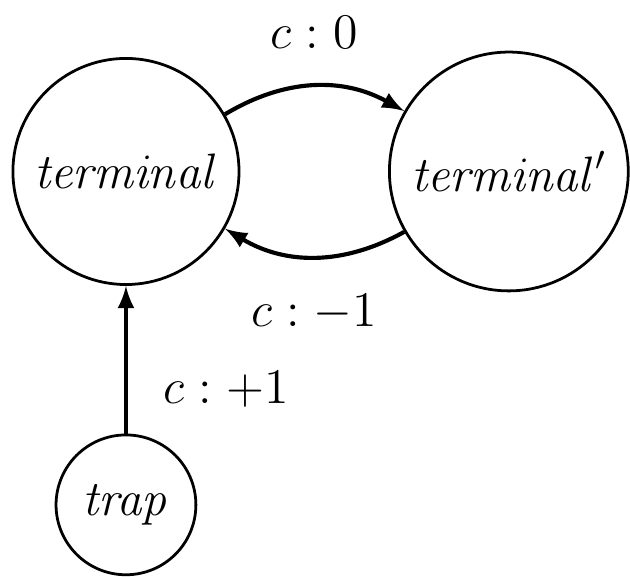}
   \end{center}

\caption[Necessary changes to the construction for the result for maximal expected termination times.]{Necessary changes to the construction for the result for for maximal expected termination times.}\label{fig:MDP_minimal_oc_time}
\end{sidefigure}

Now, let $\sched$ be a scheduler and denote the probability not to terminate before reaching $\trap$ under $\sched$ by $p^{\sched}$. For the expected termination time $T$ in $\cN$, we now have
\[
\mathbb{E}^{\sched}_{\cN,\sinit} = \left( \sum_{i=1}^{\infty} (1/2)^i(i+2(i+1))  \right) + 3\cdot p^{\sched} = 8+3\cdot p^{\sched}.
\]
The summands $(1/2)^i(i+2(i+1))$ correspond to the probability to accumulated weight $i$ in the initial component which takes $i$ steps and the $2(i+1)$ steps needed to terminate by alternatingly leaving the weight unchanged and reducing it by $1$. The three additional steps after $\trap$ occur precisely with probability $p^{\sched}$.

Not terminating before $\trap$ corresponds exactly to not terminating at all in the MDP constructed for Corollary \ref{cor:positivity_oc}. The termination probability there is hence $1-p^{\sched}$ for any scheduler~$\sched$. It is hence possible to terminate with a probability less than $\vartheta$ in that MDP if and only if it is possible to reach an expected termination time of more than $11-3\vartheta$ in $\cN$.
By Corollary \ref{cor:positivity_oc} and the fact that termination is reached almost surely in $\cN$ under any scheduler, we hence conclude:
\begin{corollary}
Let $\cM$ be a one-counter MDP with initial state $\sinit$ that terminates almost surely under any scheduler, let $\vartheta$ be a rational, and let $T$ be the random variable assigning the termination time to runs. The Positivity problem is polynomial-time reducible to the problem whether 
\[
\mathbb{E}^{\max}_{\cM,\sinit} (T) >\vartheta.
\]
\end{corollary}

The analogous argument with  similar changes to the MDP used in the proof of Theorem~\ref{thm:positivity_oc-mdp} can be used to show the analogous result for the problem whether  $\mathbb{E}^{\min}_{\cM,\sinit} (T) <\vartheta$.

\subsection{Partial and conditional stochastic shortest path problems}\label{sub:Skolem_partial}\label{sec:positivity_sspp}
\label{subsec:hardness_threshold_PE}

Our next goal is to prove that the partial and conditional SSPPs are Positivity-hard. Note that this stands in strong contrast to the classical SSPP, which is solvable in polynomial time \cite{bertsekas1991,deAlfaro1999,lics2018}.
 We start by providing a formal definition of the decision versions of these two problems.

Let $\cM$ be an MDP with a designated set of terminal states $\Goal$.
We define the random variable $\oplus\Goal$ on maximal paths $\zeta$ of $\cM$:
 \[
 \oplus \Goal (\zeta) = \begin{cases}
 \wgt(\zeta) & \text{ if }\zeta\vDash \lozenge \Goal,\\
 0 & \text{ otherwise}. 
 \end{cases}
 \]
 The objective in the \emph{partial SSPP} is to maximize the expected value of $\oplus \Goal$ which we call the \emph{partial expected accumulated weight}, or \emph{partial expectation} for short, i.e., to compute the value
 \[
{\PE}^{\max}_{\cM}\eqdef \mathbb{E}^{\max}_{\cM,\sinit}(\oplus\Goal)=\sup_\sched\mathbb{E}^{\sched}_{\cM,\sinit}(\oplus\Goal)
\]
where the supremum ranges over all schedulers $\sched$.
The threshold problem asks, given a rational~$\vartheta$, whether 
\[
{\PE}^{\max}_{\cM}>\vartheta.
\]
Note that the minimization of the partial expectation can be reduced to the maximization by multiplying all weights in $\cM$ with $-1$.  

 The \emph{conditional expectation} under a scheduler $\sched$ that reaches $\Goal$ with positive probability is the value 
\[
\CE^{\sched}_{\cM} \eqdef \mathbb{E}^{\sched}_{\cM}(\oplus\Goal \mid \lozenge \Goal).
\]
Again, we are interested in the maximal value
\[
\CE^{\max}_{\cM} \eqdef \sup_\sched \CE^{\sched}
\]
where the supremum ranges over all schedulers $\sched $ with $\Pr^{\sched}_{\cM}(\lozenge \Goal)>0$.
Consequently, the threshold problem asks for a given rational $\vartheta$ whether
\[
\CE^{\max}>\vartheta.
\]
Again, multiplying all weights with $-1$ reduces the minimization of the conditional expectation to the maximization.
Furthermore, given a further set of states $F$, the problem to maximize $\mathbb{E}^{\sched}_{\cM}(\rawdiaplus \Goal \mid \lozenge F)$ among all schedulers $\sched $ that reach $F$ with positive probability can be reduced to the conditional SSPP in our formulation as shown in \cite{tacas2017}\footnote{In \cite{tacas2017}, only MDPs with non-negative weights are considered. The  reduction of  \cite{tacas2017}, however, does not require the restriction to non-negative weights.}. 

\paragraph{Partial SSPP.}
In the sequel, we will provide a direct reduction from the Positivity problem to the partial SSPP using our modular approach via MDP-gadgets to prove the following  result:

\begin{theorem}\label{thm:positivity_PE}
The Positivity problem is polynomial-time reducible to the decision version of the partial SSPP, i.e., the question whether
\[
\PE^{\max}_{\cM}>\vartheta
\]
for a given  MDP $\cM$ and a given rational $\vartheta$.
\end{theorem}

Again, we split up the proof of the theorem into the construction of the MDP with the proof of the correctness of the encoding of the linear recurrence sequence and the computation of the threshold $\vartheta$.

\paragraph{Proof of Theorem \ref{thm:positivity_PE}: construction of the MDP and correctness of the encoding of a linear recurrence sequence.}

Let $k$ be a natural number and let $(u_n)_{n\geq 0}$ be the linear recurrence sequence given by rationals $\alpha_i$ for $1\leq i \leq k$ and $\beta_j$ for $0\leq j \leq k{-}1$ via $u_0=\beta_0$, \dots, $u_{k-1}=\beta_{k-1}$ and 
$ u_{n+k} = \alpha_1 u_{n+k-1} + \dots + \alpha_k u_n $
for all $n\geq 0$. 
By Assumption \ref{ass:1}, we can assume w.l.o.g. that   $\sum_i |\alpha_i|<\frac{1}{4}$ and that $0\leq \beta_j< \frac{1}{4k^{2k+2}}$ for all $j$.

We begin by constructing  a gadget $\cP_{\bar{\beta}}$ that encodes the initial values $\beta_0,\dots,\beta_{k{-}1}$.
The gadget  is depicted in Figure \ref{fig:gadget_initial_values} and contains states $t$, $s$, $\goal$, and $\fail$. For each $0\leq j\leq k-1$, it additionally contains states $x_j$ and $y_j$. In state $x_j$, there is one action enabled that leads to $\goal$ with probability $\frac{1}{2k^{2(k-j)}}+\beta_j$ and to $\fail$ otherwise. From state $y_j$, $\goal$ is reached with probability $\frac{1}{2k^{2(k-j)}}$ and $\fail$ otherwise. In state $t$, there is an action $\gamma_j$ leading to $x_j$ with weight $k-j$ for each $0\leq j\leq k-1$. Likewise, in state $s$ there is an action $\delta_j$ leading to $y_j$ with weight $k{-}j$ for each $0\leq j\leq k-1$.

\begin{figure}[t]
     \centering
     \includegraphics[width=0.45\linewidth]{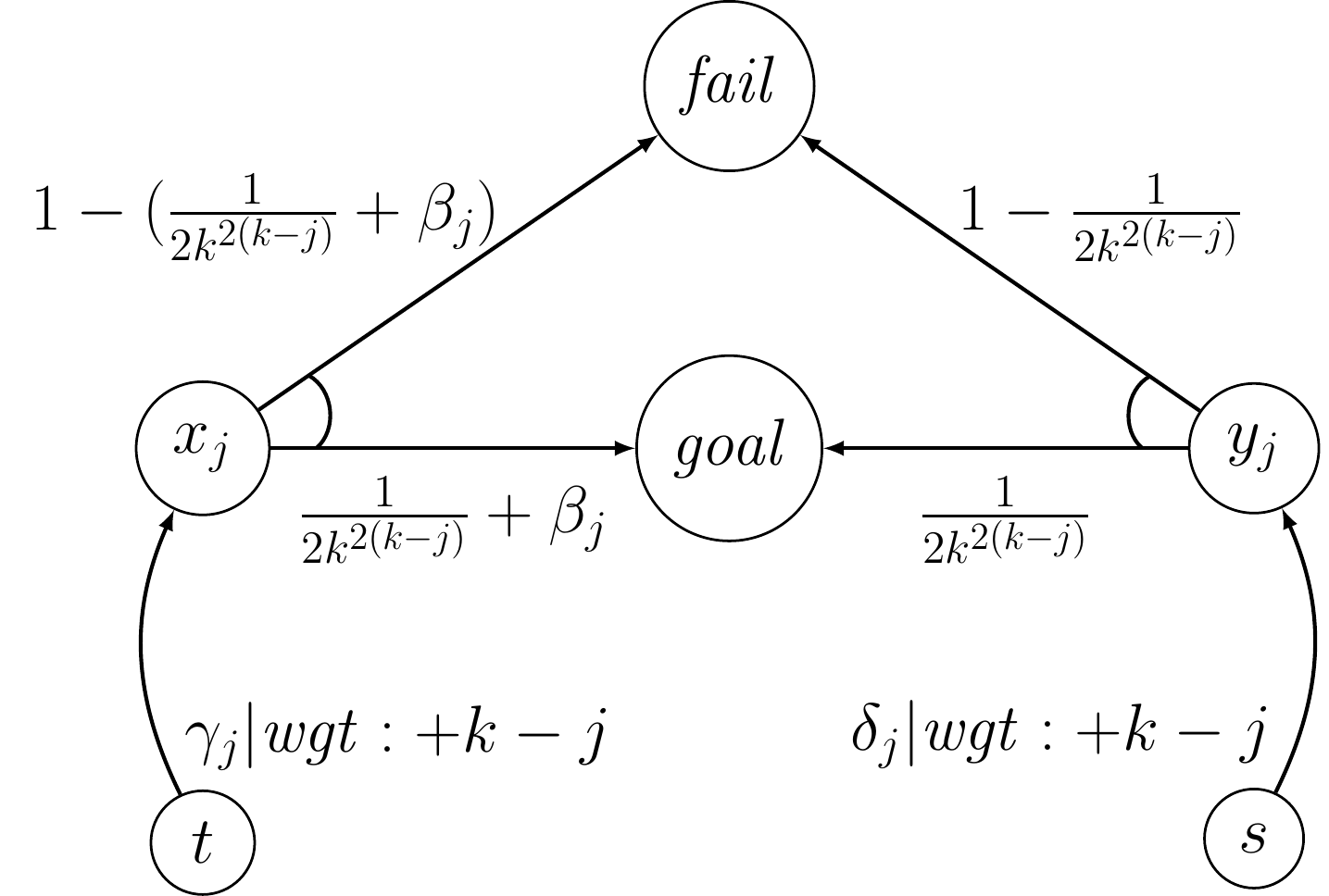}

        \caption{The gadget  $\cP_{\bar{\beta}}$ encoding  the initial values in the reduction to the threshold problem for partial expectations.    }
        \label{fig:gadget_initial_values}
\end{figure}

We furthermore  reuse the initial gadget $\cI$ and the gadget encoding the linear recurrence relation $\cG_{\bar{\alpha}}$ from the previous section.
In the gadget $\cG_{\bar{\alpha}}$, we rename the absorbing state $\trap$ to the terminal state $\goal$ which is the target state for the partial SSPP.
As before, we  glue together the three gadgets $\cI$, $\cG_{\bar{\alpha}}$ and $\cP_{\bar{\beta}}$ at states $s$, $t$, and $\goal$. 
 Let us call the full MDP that we obtain in this way $\cM$ which is depicted in  Figure \ref{fig:full_MDP}. We denote the state space by $S$.

\begin{figure}[t]
\begin{center}
\includegraphics[width=0.9\linewidth]{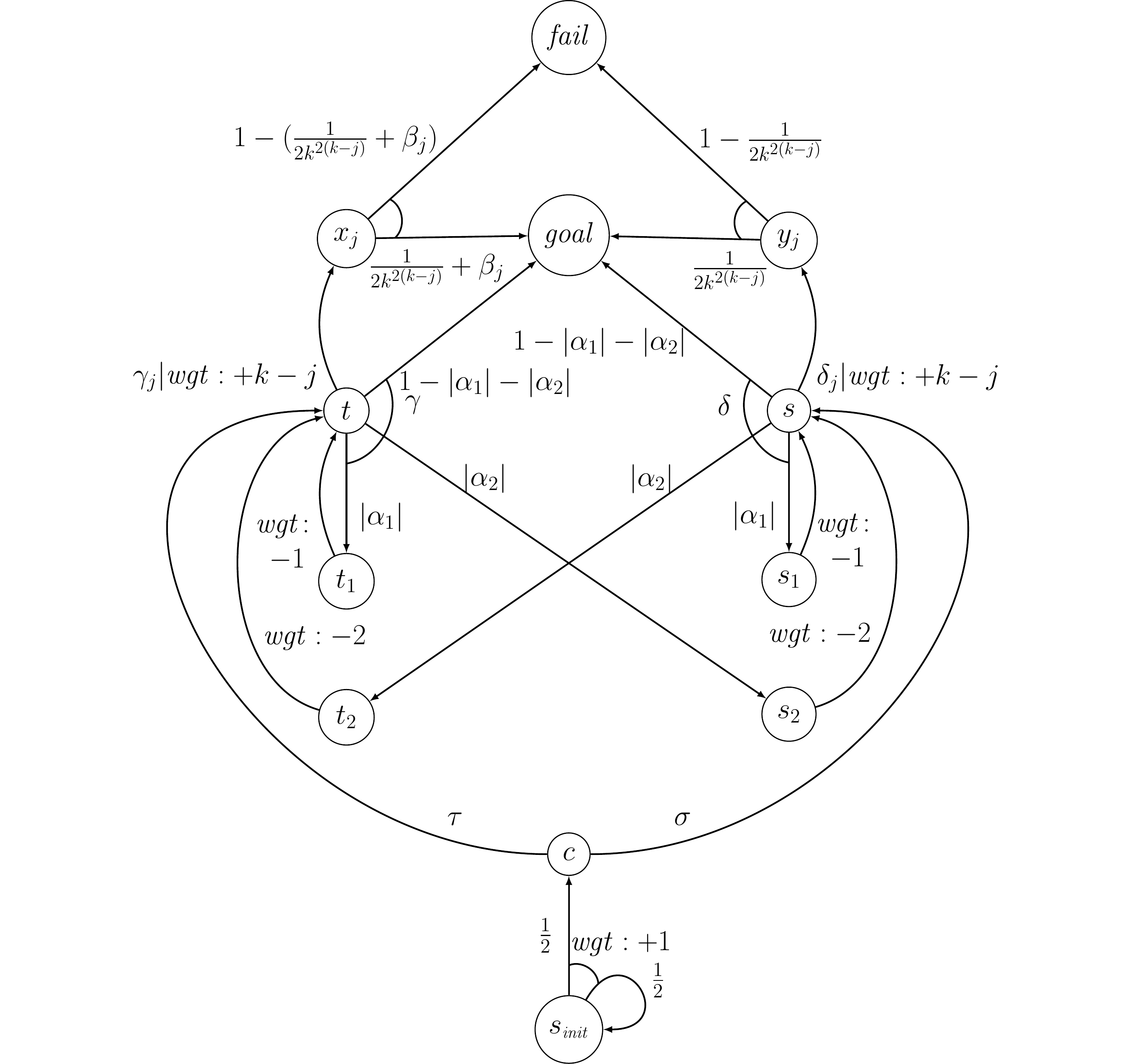}

   \end{center}

\caption[The full MDP for the Positivity-hardness proof for partial expectations.]{The full MDP for the Positivity-hardness proof for partial expectations. The MDP contains the upper part for all $0\leq j \leq k-1$. The middle part is depicted for $k=2$, $\alpha_1\geq 0$, and $\alpha_2<0$.}\label{fig:full_MDP}
\end{figure}

 The  somewhat complicated choices of probability values lead to the following lemma  showing the correct interplay between the gadgets constructed via straight-forward computations.

\begin{lemma}\label{lem:optimal_scheduler}
Consider the full MDP $\cM$. Let $0\leq j \leq k-1$. Starting with  weight $-(k{-}1){+}j$ in state $t$ or $s$,  action $\gamma_j$ and $\delta_j$  maximize the partial expectation.  For positive starting weight, $\gamma$ and $\delta$ are optimal.
\end{lemma}

\begin{proof}
Suppose action $\gamma_i$ is chosen in state $t$ when starting with weight $-(k-1)+j$. So, state $x_i$ is reached with weight  $-(k-1)+j+(k-i)=1+j-i$. Then the partial expectation achieved from this situation is 
\[(1+j-i)(\frac{1}{2k^{2(k-i)}}+\beta_i).\]
For $i>j$ this value is $\leq 0$ and hence $\gamma_i$ is certainly not optimal. For $i=j$, we obtain a partial expectation of
 \[\frac{1}{2k^{2(k-j)}}+\beta_j.\]
For $i<j$, state $x_i$ is reached with weight  $1+j-i\leq k$. Further,  $\beta_i \leq \frac{1}{4k^{2k+2}}$ and  $\frac{1}{2k^{2(k-i)}}\leq \frac{1}{2k^{2(k-j)}\cdot k^2}$. So, the partial expectation obtained via $\gamma_i$ is at most
\[ \frac{k}{2k^{2(k-j)} \cdot k^2}  + \frac{k}{4k^{2k+2}}<\frac{1}{2k^{2(k-j)}}.\]
So, indeed action $\gamma_j$  maximizes the partial expectation among the actions $\gamma_i$ with $0\leq i \leq k-1$ when the accumulated weight in state $t$ is $-(k-1)+j$.
The argument for state $s$ is the same with $\beta_i=0$ for all $i$.
It is easy to see that for accumulated weight $-(k-1)+j$ with $0\leq j\leq k-1$ actions $\gamma$ or $\delta$ are not optimal in state $t$ or $s$: If $\goal$ is reached immediately, the weight is not positive and otherwise states $t$ or $s$ are reached with lower accumulated weight again. The values $\beta_j$ are chosen small enough such that also a switch from state $t$ to $s$ while accumulating negative weight does not lead to  a higher partial expectation.

For positive accumulated weight $w$, the optimal partial expectation when choosing $\gamma$ first is at least $\frac{3}{4}w$ by construction and the fact that a positive value can be achieved from any possible successor state via one of the actions $\gamma_j$ and $\delta_j$ with $0\leq j \leq k-1$. Choosing $\gamma_j$ on the other hands results in a partial expectation of at most $(k+w)\cdot ( \frac{1}{4k^{2k+2}} + \frac{1}{2k^2})$ which is easily seen to be less than $\frac{3}{4}w$ as $k\geq 2$.
\end{proof}

For each weight $w$, denote by $e(t,w)$ and $e(s,w)$ the optimal partial expectation when starting in state $t$ or $s$ with accumulated weight $w$ in $\cM$ as if the respective state was reached from the initial state with weight $w$ and probability $1$.
For each weight $w\geq -k+1$, denote by $d(w)$ the difference $e(t,w)-e(s,w)$ between these optimal partial expectation when starting in state $t$ and $s$ with weight $w$. 
Comparing action $\gamma_j$ and $\delta_j$ for starting weight $-(k{-}1){+}j$, we conclude from the previous lemma that the difference between optimal values $d(-(k{-}1){+}j)$ is equal to $\beta_j$, for $0\leq j\leq k-1$.

The important fact we use next is that for partial expectations, the optimal values $e(r,w)$ for states $r\in S\setminus \{\goal\}$ and starting weights $w\in \mathbb{Z}$ satisfies 
the optimality equation ($\ast$) from Section \ref{sec:gadget_recurrence} when setting $e(\goal,w)=w$ as already shown in \cite{chen2013}:
\[
e(r,w)=\max_{\alpha\in \Act(r)} \sum_{r^\prime\in S} P(r,\alpha,r^\prime)\cdot e(r^\prime,w+\wgt(s,\alpha)).
\]
By the fact that $\cG_{\bar{\alpha}}$ encodes the given linear recurrence relation as soon as $\gamma$ and $\delta$ are the optimal actions as shown in Section \ref{sec:gadget_recurrence}, we conclude the following lemma:
\begin{lemma}
\label{lem:recurrence_PE}
Consider the linear recurrence sequence $(u_n)_{n\geq 0}$ given above by $\alpha_1,\dots, \alpha_k$ and $\beta_0,\dots,\beta_{k-1}$ and the MDP $\cM$ constructed from this sequence.
We have
 \[d({-}(k{-}1)+n)=u_n\] for all $n$ with the values $d(w)$ just defined.
\end{lemma}

\paragraph{Proof of Theorem \ref{thm:positivity_PE}: computation of the threshold $\vartheta$.}
Let us now consider a run of the MDP $\cM$. For any $w>0$, state $c$ is reached with accumulated weight $w$ with positive probability. As before, an optimal scheduler  has to decide whether the partial expectation when starting with weight $w$ is better in state $s$ or $t$: 
 Action $\tau$ is optimal in $c$ for accumulated weight $w$ if and only if $d(w)\geq 0$.
 Once $t$ or $s$ is reached, the optimal actions are given by Lemma \ref{lem:optimal_scheduler}.
Let ~$\sched$ be the scheduler that always chooses $\tau$ in $c$ and actions $\gamma, \gamma_0, \dots, \gamma_{k-1}, \delta, \dots$ as described in Lemma \ref{lem:optimal_scheduler}.
We conclude that $\sched$ is optimal if and only if the given linear recurrence sequence is non-negative.
The remaining step is hence in our reduction is hence to prove that the partial expectation under $\sched$ is rational and can be computed in polynomial time:

\begin{lemma}\label{thm:threshold_PE}
Let $\sched$ be the scheduler for the constructed MDP $\cM$ always choosing $\tau$ in $c$ and actions $\gamma, \gamma_0, \dots, \gamma_{k-1}, \delta, \dots$ as described in
 Lemma \ref{lem:optimal_scheduler}. The value $\PE^{\sched}_{\cM}$ is  rational and  computable in polynomial time.
\end{lemma}

\begin{proof}
Recall that the scheduler $\sched$ chooses $\gamma$ and $\delta$, respectively, as long as the accumulated weight is positive.  For an accumulated weight of $-(k-1)+j$ for $0\leq j \leq k-1$, it chooses actions~$\gamma_j$ and $\delta_j$, respectively. 

Analogously to the proof of Lemma \ref{lem:compute_threshold_OC}, we want to recursively express the partial expectations under $\sched$ starting from $t$ or $s$ with some positive accumulated weight $n\in \mathbb{N}$ which we again denote by $e(t,n)$ and $e(s,n)$, respectively.
In order to do so, we reuse the following Markov chain $\cC$ from Lemma \ref{lem:compute_threshold_OC}
also depicted in Figure \ref{fig:Markov_chain} which we briefly recall here:
The Markov chain $\cC$ has $5k$ states named $t_{-k+1}$, \dots, $t_{+k}$, $s_{-k+1}$, \dots, $s_{+k}$, and $\goal_{+1}$, \dots, $\goal_{+k}$. 
States $t_{-k+1}$, \dots, $t_{0}$, $s_{-k+1}$, \dots, $s_{0}$, and $\goal_{+1}$, \dots, $\goal_{+k}$ are absorbing.
 For $0<i,j\leq k$, there are transitions from $t_{+i}$ to $t_{+i-j}$ with probability $\alpha_j$ if $\alpha_j>0$, to $s_{+i-j}$ 
 with probability $|\alpha_j|$ if $\alpha_j<0$, and to $\goal_{+i}$ with probability $1-|\alpha_1|-\ldots-|\alpha_k|$. Transitions from $s_{+i}$ are defined analogously.

 The idea behind this Markov chain is that the reachability probabilities describe how, for arbitrary $n\in \mathbb{N}$ and $1\leq i \leq k$, 
 the values $e(t,nk+i)$ and $e(s,nk+i)$ depend on $n$ and  the values $e(t,(n-1)k+j)$ and $e(s,(n-1)k+j)$ for $1\leq j \leq k$. 
 The transitions in $\cC$ behave as $\gamma$ and $\delta$ in $\cM$, but the decrease in the accumulated weight is explicitly encoded into the state space. 
 Namely, for  $n\in\mathbb{N}$ and $0<i\leq k$, we have
{ \begin{align}\label{eqn:cC}
  e(t,nk+i) = \sum_{j=1}^k &\big(\Pr_{\cC, t_{+i}} (\lozenge t_{-k+j})\cdot e(t,(n{-}1)k+j) \nonumber 
  +\Pr_{\cC, t_{+i}} (\lozenge s_{-k+j})\cdot e(s,(n{-}1)k+j) \big) \\
  + \sum_{j=1}^k & \Pr_{\cC, t_{+i}} (\lozenge \goal_{+j}) \cdot (nk+j)    \tag{$\ast\ast$}
\end{align}}
and analogously for $ e(s,nk+i)$.
We now group the optimal values together in the following column vectors
\[v_n = (e(t,nk+k) , e(t,nk+k-1), \ldots, e(t,nk+1), e(s,nk+k) , \ldots, e(s,nk+1))^\top \]
for $n\in \mathbb{N}$. In other words, this vector contains the optimal values for the partial expectation when starting in $t$ or $s$ with an accumulated weight from  $\{nk+1,\dots,nk+k\}$. 
Further, we define the vector containing the optimal values for weights in $\{-k+1,\dots,0\}$ which are the least values of accumulated weight reachable under scheduler $\sched$.
\[v_{-1} = (e(t,0) , e(t,-1), \ldots, e(t,-k+1), e(s,0) , e(s,-1), \ldots, e(s,-k+1))^\top. \] 
As we have seen, these values are given as follows:
\[e(t,-k+1+j)=\frac{1}{2k^{2(k-j)}}+\beta_j \text{ and }e(s,-k+1+j)=\frac{1}{2k^{2(k-j)}}\] for $0\leq j \leq k-1$.

\allowdisplaybreaks

As the reachability probabilities in $\cC$ are rational and computable in polynomial time, we conclude from \eqref{eqn:cC} 
that there are a matrix $A\in \mathbb{Q}^{2k\times2k}$, and vectors $a$ and $b$ in $\mathbb{Q}^{2k}$ computable in polynomial time such that 
\[ v_n = A v_{n-1} + n a + b,\]
for all $n\in\mathbb{N}$. We claim that the following explicit representation for $n\geq -1$ satisfies this recursion:
\[ v_n = A^{n+1} v_{-1} + \sum_{i=0}^n (n-i) A^i a + \sum_{i=0}^n A^i b.\]
We show this by induction:
Clearly, this representation yields the correct value for $v_{-1}$. So, assume $v_n = A^{n+1} v_{-1} + \sum_{i=0}^n (n-i) A^i a + \sum_{i=0}^n A^i b$. Then,
\begin{align*}
v_{n+1} & = A (A^{n+1} v_{-1} +  \sum_{i=0}^n (n-i) A^i a +  \sum_{i=0}^n A^i b) +(n+1) a + b \\
& = A^{n+2} v_{-1} + \left(\sum_{i=0}^{n} (n-i) A^{i+1} a\right) + (n+1) A^0 a + \left(\sum_{i=1}^{n+1} A^i b \right) + A^0 b \\
& = A^{n+2} v_{-1} + \sum_{i=0}^{n+1} (n+1-i) A^i a + \sum_{i=0}^{n+1} A^i b.
\end{align*}
So, we have an explicit representation for $v_n$.
The value we are interested in is 
\[\PE^\sched_{\cM} = \sum_{\ell=1}^\infty (1/2)^\ell e(t,\ell).\]

\noindent
Let $c=(\frac{1}{2^k} , \frac{1}{2^{k-1}} , \dots, \frac{1}{2^1} , 0 ,\dots, 0)$. 
Then, 
\[(\frac{1}{2^k})^n c\cdot v_n= \sum_{i=1}^{k}  \frac{1}{2^{nk+i}} e(t,nk+i).\]

\noindent
Hence, we can write 
\begin{align*}
&\PE^\sched_{\cM} = \sum_{n=0}^\infty (\frac{1}{2^k})^n c \cdot v_n = c \cdot \sum_{n=0}^\infty (\frac{1}{2^k})^n v_n \\
= {} & c \cdot \sum_{n=0}^\infty (\frac{1}{2^k})^n (A^{n+1} v_{-1} + \sum_{i=0}^n (n-i) A^i a + \sum_{i=0}^n A^i b) \\
= {} & c \cdot \big(  (\sum_{n=0}^\infty (\frac{1}{2^k})^n A^{n+1}) v_{-1}  + (\sum_{n=0}^\infty (\frac{1}{2^k})^n \sum_{i=0}^n (n-i) A^i) a + (\sum_{n=0}^\infty (\frac{1}{2^k})^n \sum_{i=0}^n A^i) b              \big).
\end{align*}

We claim that all of the matrix series involved converge to rational matrices. As in the proof of Lemma \ref{lem:compute_threshold_OC}, we observe that the maximal row sum in $A$ is at most $|\alpha_1|{+}\ldots{+}|\alpha_k|<1$ because the rows of the matrix contain exactly the probabilities to reach $t_0$, \dots $t_{-k+1}$, $s_0$, \dots, and  $s_{-k+1}$ from a state $t_{+i}$ or $s_{+i}$ in $\cC$ for $1\leq i \leq k$. But the probability to reach $\goal_{+i}$ from these states is already $1{-}|\alpha_1|{-}\ldots{-}|\alpha_k|$. Hence, $\Vert A \Vert_{\infty}$, the operator norm induced by the maximum norm $\Vert \cdot \Vert_\infty$, which equals $\max_{i} \sum_{j=1}^{2k} |A_{ij}|$, is less than $1$.
So, of course also $\Vert \frac{1}{2^k} A \Vert_{\infty}<1$ and hence the  Neumann series $\sum_{n=0}^\infty (\frac{1}{2^k} A)^{n}$ converges to $(I_{2k}-\frac{1}{2^k} A)^{-1}$ where $I_{2k}$ is the identity matrix of size $2k{\times }2k$. So,
\begin{align*}
 \sum_{n=0}^\infty (\frac{1}{2^k})^n A^{n+1} = A \sum_{n=0}^\infty (\frac{1}{2^k} A)^{n} = A(I_{2k}-\frac{1}{2^k} A)^{-1}.
\end{align*}
Note that $\Vert A \Vert_\infty <1$ also implies that $I_{2k}-A$ is invertible. We observe that for all $n$,
\[\sum_{i=0}^n A^i = (I_{2k}-A)^{-1} (I_{2k} - A^{n+1}) \]
which is shown by straight-forward induction.
Therefore,
\begin{align*}
\sum_{n=0}^\infty (\frac{1}{2^k})^n \sum_{i=0}^n A^i & =  (I_{2k}-A)^{-1} \left( \sum_{n=0}^\infty (\frac{1}{2^k})^n I_{2k} - A  \sum_{n=0}^\infty (\frac{1}{2^k}A)^n   \right) \\
 & =  (I_{2k}-A)^{-1} \left(\frac{2^k}{2^k{-}1} I_{2k} - A(I_{2k}-\frac{1}{2^k} A)^{-1}\right).
\end{align*}
Finally, we show by induction that \[\sum_{i=0}^n (n-i) A^i = (I_{2k} - A)^{-2} (A^{n+1}-A+n(I_{2k}-A)).\]
This is equivalent to \[ (I_{2k} - A)^2 \sum_{i=0}^n (n-i) A^i = A^{n+1}-A+n(I_{2k}-A).\]
For $n=0$, both sides evaluate to $0$. So, we assume the claim holds for $n$.
\begin{align*}
(I_{2k} - A)^2 \sum_{i=0}^{n+1} (n+1-i) A^i ={}& (I_{2k} - A)^2 \sum_{i=0}^{n} (n-i) A^i + (I_{2k} - A)^2 \sum_{i=0}^{n}A^i \\
 \overset{\mathrm{IH}}{=}{} &  A^{n+1}-A+n(I_{2k}-A) + (I_{2k} - A)^2 \sum_{i=0}^{n}A^i \\
 = {}&  A^{n+1}-A+n(I_{2k}-A) + (I_{2k} - A)^2 (I_{2k}-A)^{-1} (I_{2k} - A^{n+1})\\
 = {}& A^{n+1}-A+n(I_{2k}-A) + I_{2k}-A - A^{n+1} + A^{n+2} \\
 = {}& A^{n+2}- A + (n+1) (I_{2k}-A) .
\end{align*}
The remaining series is the following:
\begin{align*}
\sum_{n=0}^\infty (\frac{1}{2^k})^n \sum_{i=0}^n (n-i) A^i  
={} &  \sum_{n=0}^\infty (\frac{1}{2^k})^n (I_{2k} - A)^{-2} (A^{n+1}-A+n(I_{2k}-A)) \\
={}& (I_{2k} - A)^{-2} \left(\sum_{n=0}^\infty (\frac{1}{2^k})^n A^{n+1} -  \sum_{n=0}^\infty (\frac{1}{2^k})^nA + \sum_{n=0}^\infty (\frac{1}{2^k})^n n(I_{2k}-A) \right) \\
={}&(I_{2k} - A)^{-2} \left(  A(I_{2k}-\frac{1}{2^k} A)^{-1} - \frac{2^k}{2^k{-}1}A+ \frac{2^k}{(2^k{-}1)^2}(I_{2k}-A) \right) .
\end{align*}

We conclude that all expressions in the representation of ${\PE}^\sched_{\cM}$ above are rational and computable in polynomial time. 
\end{proof}

As we have also seen before, the originally given linear recurrence sequence contains a negative member  if and only if the scheduler $\sched$ is not optimal:
If the sequence is non-negative, the scheduler $\sched$ is optimal as it behaves optimally once $t$ or $s$ has been reached and as it always moves to 
state $t$ instead of $s$ from $\mathit{choice}$ for any value $w$ of accumulated weight, which is optimal  by Lemma \ref{lem:recurrence_PE}.
If the sequence is negative at position $n$, then a scheduler that behaves like~$\sched$, except for choosing $\sigma$ for accumulated weight 
$w=n-(k{-}1)$ in state $\mathit{choice}$ is better than~$\sched$ as $e(s,w)>e(t,w)$ in this case. Note that for positions $n=0,\dots,k{-}1$, we may assume that~$u_n$ is non-negative as otherwise the sequence trivially becomes negative.
So, the originally given linear recurrence sequence contains a negative member  if and only if  
${\PE}^{\max}_{\cM} > {\PE}^\sched_{\cM}$ for the MDP $\cM$ constructed from the linear recurrence sequence in polynomial time above. 
This finishes the proof of Theorem \ref{thm:positivity_PE}.

\paragraph{Conditional SSPP.}
For the Positivity-hardness of the threshold problem for  conditional expectations, we provide a reduction from the threshold problem for  partial expectations in the following lemma. Note that a reduction in the other direction is provided in \cite{fossacs2019} rendering the two problems polynomial-time inter-reducible.

\begin{lemma}\label{lem_app:pe_ce_inter-reducible}
The threshold problem for the partial SSPP is polynomial-time reducible to the threshold problem of the  conditional SSPP. 
\end{lemma}

\begin{proof}
Let $\cM$ be an MDP with a designated terminal target state $\goal$ and let $\vartheta$ be a rational number. We construct an MDP $\cN$ such that ${\PE}^{\max}_{\cM}>\vartheta$ if and only if ${\CE}^{\max}_{\cN}>\vartheta$.
We obtain $\cN$ by adding a new initial state $\sinit^\prime$, renaming the state $\goal$ to $\goal^\prime$, and adding a new state $\goal$ to $\cM$. In $\sinit^\prime$, one action with weight $0$ is enabled leading to the old initial state $\sinit$ and to $\goal$ with probability $1/2$ each. From $\goal^\prime$ there is one new action leading to $\goal$ with probability $1$ and weight $+\vartheta$. 

Each scheduler $\sched$ for $\cM$ can be seen as a scheduler for $\cN$ and vice versa.
Now, we observe that  for  any scheduler $\sched$,  
\[\CE^\sched_\cN= \frac{1/2(\PE^\sched_\cM+\Pr_\cM^\sched(\lozenge \goal) \vartheta)}{1/2+1/2\Pr^\sched_\cM(\lozenge \goal)}=\frac{\PE^\sched_\cM+\Pr_\cM^\sched(\lozenge \goal) \vartheta}{1+\Pr_\cM^\sched(\lozenge \goal)}.\]
Hence, ${\PE}^{\max}_{\cM}>\vartheta$ if and only if ${\CE}^{\max}_{\cN}>\vartheta$.
\end{proof}

Together with the Positivity-hardness of the threshold problem for partial expectations (Theorem \ref{thm:positivity_PE}), we conclude:

\begin{theorem}\label{thm:threshold_CE}
The Positivity problem is reducible in polynomial time to the following problem: Given an MDP $\cM$ and a rational $\vartheta$, decide whether ${\CE}^{\max}_{\cM}>\vartheta$.
\end{theorem}

\paragraph{Two-sided partial SSPP.}
To conclude this section, we prove the Positivity-hardness of a two-sided version of the partial SSPP with two non-negative weight functions. 
The key idea is that, instead of using arbitrary integer weights, we can  simulate the non-monotonic behavior of the accumulated weight along a path in the partial SSPP with arbitrary weights with two non-negative weight functions.
In the definition of the random variable $\oplus \Goal$, we can replace the choice that paths not reaching $\Goal$ are assigned weight $0$  by a second weight function.
Let $\cM=(S,\Act,\Pr,\sinit, \wgt_\goal, \wgt_\fail, \goal, \fail)$ be an MDP with two designated terminal states $\goal$ and $\fail$ and two non-negative weight functions $\wgt_\goal\colon S \times \Act\to \mathbb{N}$ and $\wgt_\fail\colon S\times \Act \to \mathbb{N}$.  Assume that the probability $\Pr_{\cM,\sinit}^{\min}(\lozenge \{\goal,\fail\})=1$.
Define the following random variable~$X$ on maximal paths $\zeta$:
\[
X(\zeta)=\begin{cases}
\wgt_\goal(\zeta) & \text{if $\zeta\vDash\lozenge\goal$,} \\
\wgt_\fail(\zeta) & \text{if $\zeta\vDash\lozenge\fail$.}
\end{cases}
\]
Due to the assumption that $\goal$ or $\fail$ is reached almost surely under any scheduler, the expected value $\mathbb{E}^{\sched}_{\cM,\sinit}(X)$ is well-defined for all schedulers $\sched$ for $\cM$. 
We call the value $\mathbb{E}^{\max}_{\cM,\sinit}(X)=\sup_\sched \mathbb{E}^{\sched}_{\cM,\sinit}(X)$ the optimal \emph{two-sided partial expectation}.
We can show that the threshold problem for the two-sided partial expectation is Positivity-hard as well by a small adjustment of the construction above.

\begin{theorem}\label{thm:positivity_two-sided}
The Positivity problem is polynomial-time reducible to the following problem:
Given an MDP $\cM=(S,\Act,\Pr,\sinit, \wgt_\goal, \wgt_\fail, \goal, \fail)$  as above and a rational $\vartheta$, decide whether $\mathbb{E}^{\max}_{\cM,\sinit}(X)>\vartheta$. 
\end{theorem}

\begin{sidefigure}[t]
     \centering
    \includegraphics[width=0.4\textwidth]{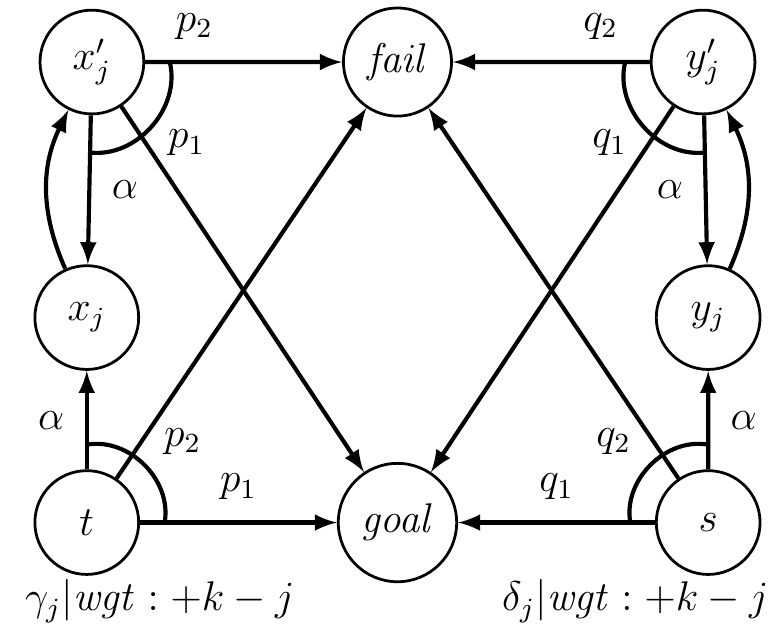}
        \caption{The  gadget $\cT_{\bar{\beta}}$ encoding  initial values in terms of two-sided partial expectations.}
        \label{fig:gadget_2_sided}
\end{sidefigure}

\begin{proof}
Given the parameters $\alpha_1,\dots, \alpha_k$ and $\beta_0,\dots,\beta_{k-1}$ of a rational linear recurrence sequence, we can construct an MDP $\cM^\prime=(S,\Act,\Pr,\sinit, \wgt, \goal, \fail)$ with one weight function $\wgt\colon S\times \Act \to \mathbb{Z}$ 
similar to the MDP $\cM$ depicted in Figure \ref{fig:full_MDP}.
W.l.o.g., we  again assume that   $\sum_i |\alpha_i|<\frac{1}{4}$ and that $0\leq \beta_j< \frac{1}{4k^{2k+2}}$ for all $j$. The non-negativity of the values
$\beta_j$ for all~$j$ can be assumed as the Positivity problem is trivial otherwise.
 The initial gadget and the gadget~$\cG_{\bar{\alpha}}$ are as before.
The gadget $\cP_{\bar{\beta}}$, however, is slightly modified and replaced by the gadget~x$\cT_{\bar{\beta}}$ depicted in Figure \ref{fig:gadget_2_sided}.
For this gadget, we define $\alpha=\sum_{i=1}^k |\alpha_i|$, $p_1=(1-\alpha) (\frac{1}{2k^{2(k-j)}}+\beta_j)$, $p_2=(1-\alpha)(1-(\frac{1}{2k^{2(k-j)}}+\beta_j))$,  $q_1=(1-\alpha) \frac{1}{2k^{2(k-j)}}$, and $q_2=(1-\alpha)(1-\frac{1}{2k^{2(k-j)}})$.
With the transitions as in the figure, the probability  to reach $\goal$ or $\fail$ and the weight accumulated does not change when choosing action $\gamma_j$ or $\delta_j$ compared to the gadget $\cP_{\bar{\beta}}$.
The only difference is that the expected time to reach $\goal$ or $\fail$ changes. 
The steps alternate between probability $1-\alpha$ and probability $0$ to reach $\goal$ or $\fail$ --  just as in the gadget $\cG_{\bar{\alpha}}$. In this way, it makes no difference for the expected time before reaching $\goal$ or $\fail$ when a scheduler stops choosing $\gamma$ and $\delta$.
We can, in fact, compute the expected time $T$ to reach $\goal$ or $\fail$ from $\sinit$ under any scheduler quite easily:
Reaching $t$ or $s$ takes $3$ steps in expectation. 
Afterwards, the number of steps taken is $1+2\ell$ with probability $\alpha^\ell\cdot (1-\alpha)$. In expectation, 
this yields
\[
\sum_{\ell=0}^\infty (1+2\ell) \alpha^\ell\cdot (1-\alpha) = \left( \sum_{\ell=0}^\infty 2(\ell+1) \alpha^\ell\cdot (1-\alpha)  \right) -1 = \frac{2}{1-\alpha} -1.
\]
additional steps. So,
\[
T=3+\frac{2}{1-\alpha}-1 = 2+\frac{2}{1-\alpha} .
\] 
The optimal scheduler $\sched$ for the partial expectation in $\cM^\prime$ is the same as in the MDP $\cM$ above. Also, the value $\vartheta$ of this scheduler can be computed as in Lemma \ref{thm:threshold_PE}.
So, $\PE^{\max}_{\cM^{\prime},\sinit} > \vartheta$ if and only if the given linear recurrence sequence is eventually negative.

Note that all weights in $\cM^\prime$ are $\geq -k$.
We define two new weight functions to obtain an MDP $\cN$ from $\cM^\prime$: We let $\wgt_\goal(s,\alpha)=\wgt(s,\alpha)+k$ and $\wgt_\fail(s,\alpha)=+k$ for all $(s,\alpha)\in S\times\Act$. Both weight functions take only non-negative integer values. 

Any scheduler $\sched$ for $\cM^\prime$ can be viewed as a scheduler for $\cN$, and vice versa, as the two MDPs only differ in the weight functions. 
Further, we observe that for each maximal path $\zeta$ ending in $\goal$ or $\fail$ in $\cM^\prime$ and at the same time in $\cN$, we have $X(\zeta)=\oplus \goal (\zeta)+k\cdot \length(\zeta)$.
(Recall that $\oplus \goal (\zeta)$ equals $\wgt(\zeta)$ if $\zeta$ reaches $\goal$ and $0$ if $\zeta$ reaches $\fail$.)
As the expected time before $\goal$ or $\fail$ is reached is constant, namely $T$ under any scheduler, it follows that for all schedulers $\tsched$ we have
\[
\mathbb{E}^{\tsched}_{\cN,\sinit}(X) = \PE^{\tsched}_{\cM^\prime,\sinit} + k\cdot T.
\]
Therefore, $\mathbb{E}^{\max}_{\cN,\sinit} (X)>\vartheta+ k\cdot T$ if and only if the given linear recurrence sequence eventually becomes negative. 
\end{proof}

\subsection{Conditional value-at-risk for accumulated weights}\label{sec:positivity_cvar}

Lastly, we aim to prove the Positivity-hardness of the threshold problem for the CVaR in this section. 
To this end, we provide a further direct reduction from the Positivity-problem to the threshold problem for the expected value of an auxiliary random variable closely related to the CVaR using our MDP-gadgets.

\paragraph{Conditional Value-at-Risk.}
Given an MDP $\cM=(S,\Act,P,\sinit,\wgt,\Goal)$ with a scheduler $\sched$, a random variable $X$ defined on runs of the MDP with values in $\mathbb{R}$ and a value $p\in [0,1]$, we define the {value-at-risk} as 
$\VaR^{\sched}_{p}(X) = \sup \{r\in \mathbb{R}| \Pr_{\cM}^\sched (X\leq r)\leq p\}$.
So, the value-at-risk is the point at which the cumulative distribution function of $X$ reaches or exceeds $p$. 
The {conditional value-at-risk} is now the expectation of $X$ under the condition that the outcome belongs to the $p$ worst outcomes -- in this case, the $p$ lowest outcomes.
Denote $\VaR_p^\sched(X)$ by $v$. 
Following the treatment of  random variables that are not continuous in general in \cite{kretinsky2018}, we define the conditional value-at-risk as follows:
\[\CVaR_p^\sched(X) = 1/p (  \Pr_{\cM}^{\sched}(X<v) \cdot \mathbb{E}_{\cM}^\sched(X|X<v) + (p- \Pr_{\cM}^{\sched}(X<v))\cdot v  ).\]
Outcomes of $X$ which are less than $v$  are treated differently to outcomes equal to $v$ as it is possible that the outcome $v$ has positive probability and we only want to 
 account exactly for the $p$ worst outcomes. Hence,  we take only $p- \Pr_{\cM}^{\sched}(X<v)$ of the outcomes which are exactly $v$ into account as well.
 To provide worst-case guarantees or to find risk-averse policies, we are  interested in the maximal and minimal conditional value-at-risk 
 \[\CVaR_p^{\max}(X)=\sup_\sched \CVaR_p^\sched(X) \text{  and }\CVaR_p^{\min}(X)=\inf_\sched \CVaR_p^\sched(X).\]
 In our formulation here, low outcomes are considered to be bad. Completely analogously, one can define the conditional value-at-risk for the highest $p$ outcomes.

The main result of the section is the following:

\begin{theorem}\label{thm:positivity_cvar}
The Positivity problem is polynomial-time reducible to the following problem:
Given an MDP $\cM$ and  rationals $\vartheta$ and  $p\in(0,1)$, decide whether \[\CVaR^{\max}_{p} (\rawdiaplus \goal)> \vartheta.\]
\end{theorem}

We will use an auxiliary optimization problem to prove this result.
We begin with the following consideration: Given an MDP $\cM$ with initial state $\sinit$,
we construct a new MDP $\cN$.
We add a new initial state $\sinit^\prime$. In $\sinit^\prime$, there is only one action with weight $0$  enabled 
leading to $\sinit$ with probability $\frac{1}{3}$ and to $\goal$ with probability $\frac{2}{3}$. So, at least two thirds of the paths accumulate weight $0$ before reaching the goal. 
Hence, we can already say that $\VaR^\sched_{1/2}(\rawdiaplus \goal)=0$ in $\cN$ under any scheduler $\sched$. Note that schedulers for $\cM$ can be seen as schedulers for $\cN$ and vice versa.
This considerably simplifies the computation of the conditional value-at-risk in $\cN$. Define the random variable $\rawdiaminus \goal $ on paths $\zeta$ by
\[
\rawdiaminus \goal (\zeta)
=\min( \rawdiaplus \goal (\zeta),0).
\]
Now, the conditional value-at-risk for the probability value $1/2$ under a scheduler $\sched$ in $\cN$ is  given by 
$
\CVaR^\sched_{1/2}(\rawdiaplus \goal)= 2 \cdot \mathbb{E}^\sched_{\cN,\sinit}(\rawdiaminus \goal) = \frac{2}{3}\cdot \mathbb{E}^\sched_{\cM,\sinit}(\rawdiaminus \goal) 
$.
So, the result follows from the following lemma:

\begin{lemma}\label{lem:Skolem_rawdiaminus}
The Positivity problem is polynomial-time reducible to the following problem:
Given an MDP $\cM$ and a rational $\vartheta$, decide whether $\mathbb{E}^{\max}_{\cM,\sinit} (\rawdiaminus \goal)> \vartheta$.
\end{lemma}

\begin{proof}
The first important observation is that the optimal expectation $e(q,w)$ of $\rawdiaminus \goal$ for different starting states $q$ and starting weights $w$ satisfies equation ($\ast$) from Section \ref{sec:gadget_recurrence}, i.e.,
$e(q,w)=\sum_{r\in S} P(q,\alpha,r)\cdot e(r,w{+}\wgt(q,\alpha))$ if an optimal scheduler chooses actions $\alpha$ in state $q\not = \goal$ when the accumulated weight is $w$. The value $e(\goal,w)$ is $w$ if $w\leq 0$ and $0$ otherwise.
This allows us to reuse the gadget $\cG_{\bar{\alpha}}$ to encode a linear recurrence relation.

We again adjust the gadget encoding the initial values of a linear recurrence sequence.
So, let $k$ be a natural number, $\alpha_1,\dots,\alpha_k$ be rational coefficients of a linear recurrence sequence, and $\beta_0,\dots, \beta_{k-1}\geq 0$ the rational initial values. W.l.o.g. we again assume these values to be small using Assumption \ref{ass:1}, namely: $\sum_{1\leq i\leq k} |\alpha_i|\leq \frac{1}{5(k+1)}$ and for all $j$, $\beta_j\leq \frac{1}{3}\alpha$ where $\alpha=\sum_{1\leq i\leq k} |\alpha_i|$.

The new gadget that encodes the initial values of a linear recurrence sequence  is depicted in Figure \ref{fig:gadget_cvar}. In states $t$ and $s$, there is a choice between actions $\gamma_j$ and $\delta_j$, respectively, for $0\leq j \leq k-1$. 
After gluing together this gadget with the gadget $\cG_{\bar{\alpha}}$ at states $t$, $s$, and $\goal$,
we prove that the interplay between the gadgets is correct:
Let $0\leq j \leq k-1$. Starting with accumulated weight ${-}k{+}j$ in state $t$, the action $\gamma_j$ maximizes the partial expectation among the actions $\gamma_0,\dots,\gamma_{k-1}$. Likewise, $\delta_j$ is optimal when starting in $s$ with weight ${-}k{+}j$. If the accumulated weight is non-negative in state $s$ or $t$, then $\gamma$ or $\delta$ are optimal. The idea is that for positive starting weights, the tail loss of $\gamma_i$ and $\delta_i$ is relatively high while for weights just below~$0$, the chance to reach $\goal$ with positive weight again outweighs this tail loss.

\begin{figure}[t]
     \centering
            \includegraphics[width=0.8\linewidth]{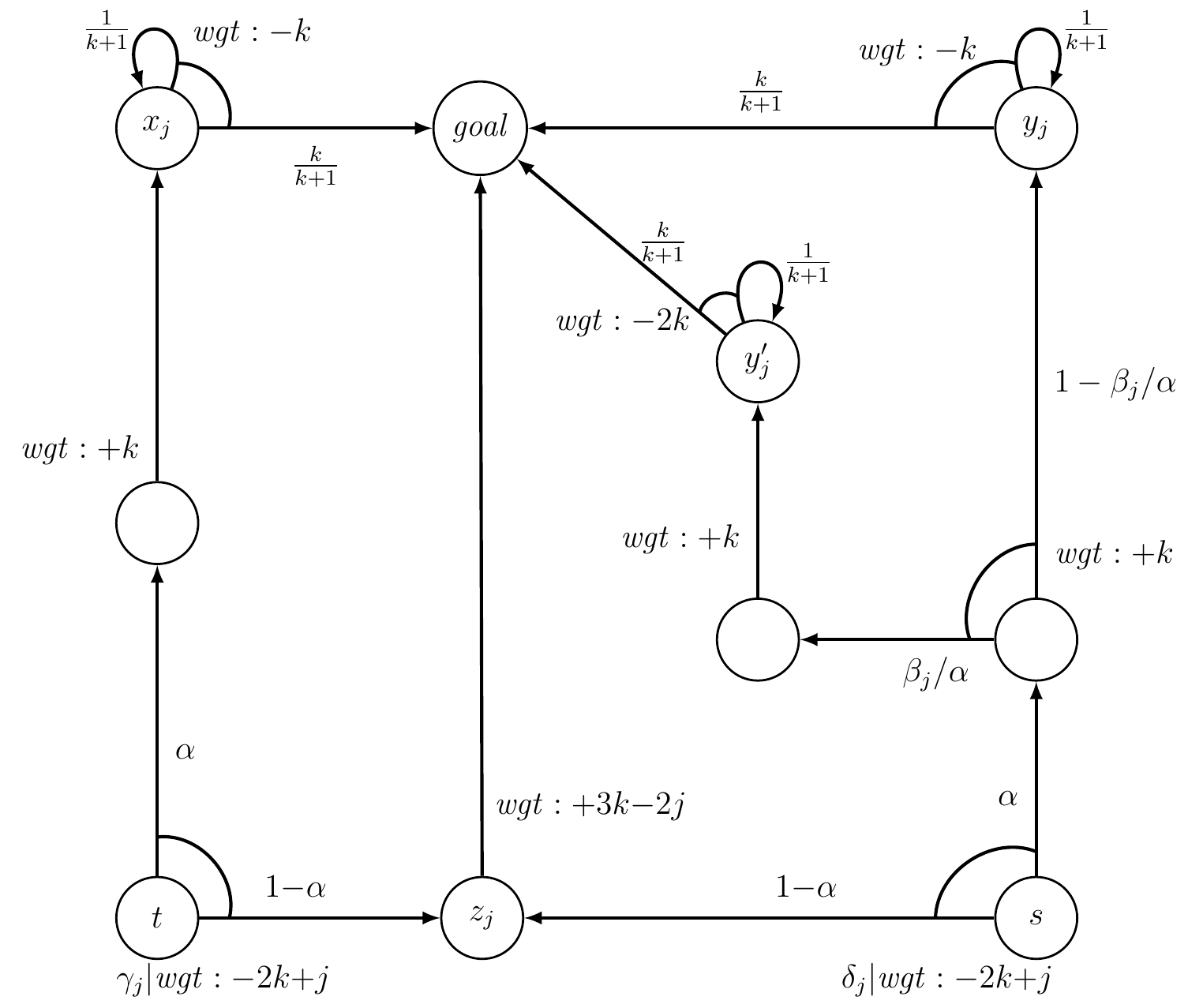}
 
        \caption[The gadget encoding initial values for the reduction to the threshold problem for the conditional value-at-risk.]{The gadget encoding initial values for the reduction to the threshold problem for the conditional value-at-risk. The gadget contains the depicted states and actions for each $0\leq j \leq k-1$.
The probability $\alpha$ is $\sum_{1\leq i \leq k} |\alpha_i|$.}
        \label{fig:gadget_cvar}
\end{figure}

First, we estimate the expectation of $\rawdiaminus \goal$ when choosing $\delta_i$ and $\delta$ while the accumulated weight is ${-}k{+}j$ in $s$. 
If $i> j$, then $\delta_i$ and $\delta$ lead to $\goal$ directly with probability $1{-}\alpha$ and weight $\leq -1$. 
So, the expectation is less than ${-}(1-\alpha)\leq {-}1{+}\frac{1}{5(k{+}1)}$. 

If $i\leq j$, then with probability $1{-}\alpha$ $\goal$ is reached with positive weight, hence $\rawdiaminus \goal$ is $0$ on these paths. 

With probability $\beta_i$, goal is reached via $y_j^\prime$. In this case all runs reach $\goal$ with negative weight. On the way to $y_j^\prime$ weight $2k$ is added, but afterwards subtracted again at least once. 
In expectation weight $2k$ is subtracted $\frac{k{+}1}{k}$ many times. Furthermore, ${-}2k{+}i$ is added to the starting weight of ${-}k{+}j$. So, these paths contribute $\beta_i\cdot(2k-2k\frac{k{+}1}{k}{-}3k{+}j{+}i)=({-}3k{+}j{+}i{-}2)\cdot\beta_i$ to the expectation of $\rawdiaminus \goal$. 

The remaining paths reach $\goal$ via $y_j$ and all reach $\goal$ with negative weight as well.
The probability to reach $y_j$ is $\alpha-\beta_j$. On the way to $y_j$, the initial weight of ${-}k{+}i$ is changed to ${-}2k{+}j{+}i$.
Afterwards, weight $-k$ is accumulated $\frac{k{+}1}{k}$-many times in expectation.
So,  these remaining paths contribute $({-}3k{+}j{+}i{-}1)\cdot(\alpha-\beta_i)$. So, all in all the expectation of $\rawdiaminus \goal$ in this situation is $\alpha {\cdot} ({-}3k{+}j{+}i{-}1){-}\beta_i$.
Now, as $\alpha\leq \frac{1}{5(k{+}1)}$ and $\beta_i\leq \frac{\alpha}{3}$ for all $i$,
we see that $\alpha {\cdot} ({-}3k{+}j{+}i{-}1){-}\beta_i\geq {-}(3k+2)\alpha\geq {-}1{+}\frac{1}{5(k{+}1)}$. The optimum with $i\leq j$ is obtained for $i=j$ as $\beta_i\leq \alpha/3$ for all $i$.
Hence
 indeed $\delta_j$ is the optimal action. For $\gamma_j$ the same proof with $\beta_i=0$ for all $i$ leads to the same result.

Now assume that the accumulated weight in $t$ or $s$ is $\ell\geq 0$. 
Then, all actions lead to $\goal$ with a positive weight with probability $1-\alpha$. 
In this case $\rawdiaminus \goal$ is $0$. 
However, a scheduler $\sched$ which always chooses $\gamma$ and $\delta$ is better than a scheduler choosing $\gamma_j$ or $\delta_j$ for any $j\leq k{-}1$. 
Under scheduler $\sched$ starting from $s$ or $t$ a run returns to $\{s,t\}$ with probability $\alpha$ while accumulating weight 
$\geq {-}k$ and the process is repeated. After choosing $\gamma_j$ or $\delta_j$ the run moves to $x_j$, $y_j$ or $y_j^\prime$ while accumulating a negative weight. From then on, in each step it will stay in that state with probability greater than $\alpha$ and accumulate weight $\leq {-}k$. Hence, the expectation of $\rawdiaminus \goal$ is lower under $\gamma_j$ or $\delta_j$ than under $\sched$. Therefore indeed $\gamma$ and $\delta$ are the best actions for non-negative accumulated weight in states $s$ and $t$.

Let now $e(t,w)$ and $e(s,w)$ denote the optimal expectations of $\rawdiaminus \goal$ when starting in $t$ or $s$ with weight $w$.
Further, let $d(w)=e(t,w)-e(s,w)$.
From the argument above, we also learn that
the difference $d(-k{+}j)$ is equal to $\beta_j$, for $0\leq j\leq k-1$ .
Put together with the linear recurrence encoded in $\cG_{\bar{\alpha}}$  this shows that $d({-k}+w) = u_w$ for all $w$ where $(u_n)_{n\in\mathbb{N}}$ is the linear recurrence sequence specified by the $\alpha_i$, $\beta_j$, $1\leq i\leq k$, and $0\leq j \leq k{-}1$. 

Finally, we add the same initial component as in the previous section to obtain an MDP~$\cM$. Let  $\sched$ be the scheduler always choosing $\tau$ in state $c$  and afterwards following the optimal actions as described above is optimal if and only if the linear recurrence sequence stays non-negative.
The remaining argument goes completely analogously to the proof of Theorem \ref{thm:positivity_oc-mdp}.
Grouping together the optimal values in vectors $v_n$ with $2k$ entries as done there, we can use the same Markov chain as in that proof to obtain a matrix $A$ such that $v_{n+1}=Av_n$. This allows us to compute the rational value $\vartheta=\mathbb{E}_{\cM,\sinit}^{\sched}(\rawdiaminus \goal)$ via a matrix series in polynomial time and 
$\mathbb{E}_{\cM,\sinit}^{\max}(\rawdiaminus \goal) > \vartheta$ if and only if the given linear recurrence sequence is eventually negative.
\end{proof}

By the discussion above, this lemma directly implies Theorem \ref{thm:positivity_cvar}.
With adaptions similar to the previous section, it is possible to obtain the analogous result for the minimal expectation of $\rawdiaminus \goal$. This implies that also the threshold problem  whether the minimal
conditional value-at-risk is less than a threshold $\vartheta$, $\CVaR^{\min}_{p} (\rawdiaplus \goal)<\vartheta$, is Positivity-hard.

\section{Conclusion}
\label{sec:conclusion}

The Positivity-hardness results established in this paper show that a series of problems on finite-state MDPs that have been studied and left open in the literature exhibit an inherent mathematical difficulty. 
A decidability result for any of these problems would imply a major break-through in analytic number theory.
At the heart of our Positivity-hardness proofs lies the construction of modular MDPs consisting of three gadgets. This construction provides a versatile proof strategy to establish Positivity-hardness results: It allowed us to provide  three direct reductions from the Positivity problem by constructing structurally identical MDPs that only differ in the gadget encoding the initial values. The further chains of reductions depicted in Figure \ref{fig:overview_positivity} established Positivity-hardness for a landscape of different problems on one-counter MDPs and integer-weighted MDPs.

The proof technique might be applicable to further threshold problems associated to optimization problems on MDPs. A main requirement for the direct applicability of the technique is that the optimal values $V(s,w)$ in terms of the current state $s$ and the weight $w$ accumulated so far, or a similar quantity that can be increased and decreased, satisfy an optimality equation of the form 
\[
V(s,w)=\max_{\alpha\in \Act(s)} \sum_{t\in S} P(s,\alpha,t)\cdot V(t,w+\wgt(s,\alpha)).
\]
In addition, the optimum must not be achievable with memoryless schedulers, but the optimal decisions have to depend on the accumulated weight to make it possible to encode initial values of a linear recurrence sequence. This combination of conditions is quite common as we have seen.

Furthermore,  our  Positivity-hardness results can be used to establish Positivity-hardness of further decision problems on MDPs, which are on first sight of a rather different nature:
In \cite{icalp2020,piribauer2021}, it is shown how our proof of the Positivity-hardness of the two-sided partial SSPP can be modified to prove the Positivity-hardness of two problems concerning the long-run satisfaction of path properties, namely the threshold problem for long-run probabilities and the model-checking problem of frequency-LTL. Both of these problems address the degree to which a property is satisfied by the sequence of suffixes of a run in order to analyze the long-run behavior of systems.
The long-run probability of a property $\varphi$ in an MDP $\cM$ under a scheduler~$\sched$ is the expected long-run average of the probability that a suffix generated by $\sched$ in $\cM$ satisfies~$\varphi$. Similarly, frequency-LTL extends LTL by an operator that requires a certain percentage of the suffixes of a run to satisfy a property.
Long-run probabilities and frequency-LTL in MDPs have been investigated in \cite{lics2019} and \cite{ForejtK15,ForejtKK15}, respectively, where decidable special cases of the mentioned decision problems have been identified. In general, however, the decidability status of these problems is open. 
The reductions in \cite{icalp2020,piribauer2021} show how the two-sided partial SSPP can be encoded into the long-run probability as well as the long-run frequency of the satisfaction of a simple regular co-safety property, i.e., the negation of a safety property, yielding Positivity-hardness for the threshold problem for long-run probabilities and the model-checking problem of frequency-LTL in MDPs.

It is worth mentioning that in the special case of Markov chains, several of the  problems investigated here are decidable: In Markov chains, partial and conditional expectations 
can be computed in polynomial time \cite{fossacs2019}.
Furthermore, one-counter Markov chains constitute a special case of recursive Markov chains, for which the threshold problem for the termination probability can be decided in polynomial space \cite{DBLP:journals/jacm/EtessamiY09}.
Remarkably however,  the threshold problem for the probability that the accumulated cost satisfies a Boolean combination of inequality constraints  in finite-state Markov chains is open \cite{haase2017computing}. 

Finally, the Positivity-hardness results  leave the possibility open that
some or all of the problems we studied are in fact harder than the Positivity problem. In particular, it could be the case that the problems are undecidable and that a proof of the undecidability would yield no implications for the Positivity problem.
For this reason, investigating whether some or all of the threshold problems are reducible to the Positivity problem constitutes a very interesting -- and challenging -- direction for future work.
Such an inter-reducibility result would show that studying any of the discussed optimization problems on MDPs could be a worthwhile direction of research to settle the decidability status of the Positivity-problem.
Some hope for an inter-reducibility result can be drawn from the fact that the optimal values are approximable for several of the problems -- for termination probabilities and expected termination times of one-counter MDPs, this was shown in \cite{brazdil2011,brazdil2012} and  for partial and conditional expectations in \cite{fossacs2019}. This indicates that there is at least a major difference to undecidable problems in a similar context such as the emptiness problem for probabilistic finite automata where the optimal value cannot be approximated \cite{paz1971,condon1989complexity}.

	\printbibliography

\end{document}